\documentclass[11pt,letterpaper]{article}

\usepackage{amssymb}
\usepackage{latexsym}
\usepackage{amsmath}
\usepackage{amsthm}
\usepackage{mathdots}
\usepackage{enumitem}

\usepackage{subcaption}
\usepackage{tikz}
\usepackage[margin=2.7cm]{geometry}

\usepackage{cite}

\numberwithin{equation}{section}
\numberwithin{figure}{section}

\theoremstyle{definition}
\newtheorem{theo}{Theorem}[section]
\newtheorem{prop}[theo]{Proposition}
\newtheorem{cor}[theo]{Corollary}
\newtheorem{defi}[theo]{Definition}
\newtheorem{exa}[theo]{Example}
\newtheorem{rem}[theo]{Remark}
\numberwithin{equation}{section}

\newtheorem{lemma}[theo]{Lemma}
\newcommand{\numberset}{\mathbb}
\newcommand{\N}{\numberset{N}}
\newcommand{\Z}{\numberset{Z}}

\newcommand{\R}{\numberset{R}}
\newcommand{\C}{\numberset{C}}

\newcommand{\F}{\numberset{F}}
\newcommand{\A}{\numberset{A}}

\newcommand{\mL}{\mathcal{L}}
\newcommand{\mC}{\mathcal{C}}
\newcommand{\mP}{\mathcal{P}}

\newcommand{\drk}{d_{\textnormal{rk}}}

\newcommand{\rs}{\textnormal{rs}}

\newcommand{\tr}{\textnormal{tr}}
\newcommand{\Tr}{\textnormal{Tr}}
\newcommand{\rk}{\textnormal{rk}}
\newcommand{\RRE}{\textnormal{RREF}}

\newcommand{\piv}{\textnormal{piv}}
\newcommand{\rpiv}{\textnormal{rpiv}}
\newcommand{\supp}{\textnormal{supp}}

\newcommand{\cC}{{\mathcal C}}

\newcommand{\cF}{{\mathcal F}}
\newcommand{\cG}{{\mathcal G}}

\newcommand{\cI}{{\mathcal I}}
\newcommand{\cJ}{{\mathcal J}}
\newcommand{\cL}{{\mathcal L}}

\newcommand{\cO}{{\mathcal O}}
\newcommand{\cP}{{\mathcal P}}
\newcommand{\cQ}{{\mathcal Q}}

\newcommand{\cS}{{\mathcal S}}
\newcommand{\cT}{{\mathcal T}}
\newcommand{\cU}{{\mathcal U}}

\newcommand{\mat}{\F^{n \times m}}

\newcommand{\len}[1]{{|\!|{#1}|\!|}}
\newcommand{\widesim}[1][1.5]{\scalebox{#1}[1]{$\sim$}}
\newcommand{\Fnm}{{\F^{n\times m}}}
\newcommand{\wcP}{\widehat{\mathcal P}}
\newcommand{\wwcP}{\widehat{\phantom{\big|}\hspace*{.5em}}\hspace*{-.9em}\wcP}

\newcommand{\GL}{\mbox{\rm GL}}
\newcommand{\inner}[1]{\langle{#1}\rangle}
\newcommand{\Binom}[2]{\genfrac{(}{)}{0pt}{1}{#1}{#2}}
\newcommand{\BinomS}[2]{\genfrac{(}{)}{0pt}{2}{#1}{#2}}
\newcommand{\Gaussian}[2]{\genfrac{[}{]}{0pt}{1}{#1}{#2}}
\newcommand{\GaussianD}[2]{\genfrac{[}{]}{0pt}{0}{#1}{#2}}

\newcommand{\SmallMat}[2]{\mbox{$\left(\begin{smallmatrix}{#1}\\{#2}\end{smallmatrix}\right)$}}

\usepackage{hyperref}

\usepackage{authblk}

\begin{document}

\title{\textbf{Partitions of Matrix Spaces\\ With an Application to $q$-Rook Polynomials}}

\author[1]{Heide Gluesing-Luerssen\thanks{The author was partially supported by grant \# 422479 from the Simons Foundation.}}
\author[2]{Alberto Ravagnani\thanks{The author was supported by the Swiss National Science Foundation through grant \# P2NEP2\_168527 and by the
 Marie Curie Research Grants Scheme, grant \# 740880.}}

\affil[1]{Department of Mathematics, University of Kentucky}
\affil[2]{School of Mathematics and Statistics, University College Dublin, Ireland}

\setlist[itemize]{noitemsep, nolistsep}

\renewcommand\Authands{ and }

\date{}

\maketitle

\thispagestyle{empty}

\begin{abstract}
We study the row-space partition and the pivot partition on the matrix space $\F_q^{n \times m}$. We show that both these partitions are reflexive and that the row-space partition is self-dual. Moreover, using various combinatorial methods, we explicitly compute the Krawtchouk coefficients associated with these partitions. This establishes MacWilliams-type identities for the row-space and pivot enumerators of linear rank-metric codes.
We then generalize the Singleton-like bound for rank-metric codes, and introduce two new concepts of code extremality.  
Both of them generalize the notion of MRD codes and are preserved by trace-duality. 
Moreover, codes that are extremal according to either notion satisfy strong rigidity properties analogous to those of MRD codes.
As an application of our results to combinatorics, we give closed formulas for the $q$-rook polynomials associated with Ferrers diagram boards.
Moreover, we exploit connections between matrices over finite fields and rook placements to prove that the number of matrices
of rank $r$ over $\F_q$ supported on a Ferrers diagram is a polynomial in~$q$,
whose degree is strictly increasing in $r$.
Finally, we investigate the natural analogues of the MacWilliams Extension Theorem for the rank, the row-space, and the pivot partitions.
\end{abstract}



\section*{Introduction}\label{sec:intro}

This paper investigates the mathematical structure of rank-metric codes,
with a particular focus on partition enumerators and their connection with the theory of $q$-rook polynomials. A rank-metric code is an $\F_q$-linear space of matrices endowed with the rank distance. The latter measures the distance between two matrices as the rank of their difference. Rank-metric codes were first studied by Delsarte~\cite{Del78} for combinatorial interest via association schemes, and were independently re-discovered by Gabidulin~\cite{Gab85}, Roth~\cite{roth1991maximum},
and Cooperstein~\cite{cooperstein1998external}
in different contexts.

In 2008, rank-metric codes were proposed as a solution to the problem of error amplification in communication networks by
Silva/K\"otter/Kschischang~\cite{SKK08}. Since then, the mathematical theory of rank-metric codes has seen a resurgence of interest. In particular, $\F_q$-linear spaces of matrices have been studied in connection with various topics in enumerative and algebraic combinatorics; see \cite{lewis2018rook,Ra18,jurrius2018defining,gorla2018rank,
shiromoto2019codes,schmidt2018hermitian,ghorpade2019polymatroid} among many others.
This paper belongs to the latter line of research.

The best known class of rank-metric codes are the \textit{maximum rank distance} (MRD) codes.
Fixing the desired matrix size, field size, and minimum rank distance, they have the largest cardinality meeting these
parameters, see~\eqref{e-Singl} and the paragraph thereafter.
MRD codes have remarkable rigidity properties:
(a) the dual of an MRD code, with respect to a natural bilinear form, is an MRD code again;
(b) the \textit{rank distribution} of an MRD code (i.e., the number of matrices of each rank) is fully determined by the parameters of the code and therefore does not depend on the particular choice of the code.
This shows the analogy between MRD codes in the rank metric and MDS codes in the classical Hamming metric, which enjoy similar properties; see \cite{HP03} for a general reference on Hamming-metric codes.

The rank distribution is a special instance of a \textit{partition distribution}. Partitioning the entire matrix space, say $\F_q^{n \times m}$, into subsets according to some property (such as the rank) gives rise to the \textit{partition enumerator} of a code, which simply encodes the number of codewords in any partition block.
In this paper, we study of the \textit{row-space partition}~$\cP^\rs$ and the \textit{pivot partition}~$\cP^\piv$ on $\F_q^{n \times m}$, and the connections between these and topics in $q$-rook theory.

In~$\cP^\rs$, matrices in $\F_q^{n\times m}$ are in the same partition block if they have the same row space, while in~$\cP^\piv$ they are grouped according to their pivot indices after row reduction.
Thus~$\cP^\rs$ is finer than $\cP^\piv$, which is finer than the rank partition,~$\cP^\rk$.
Continuing the analogy from above, where the rank distribution is the analogue of the \textit{Hamming weight distribution},
the row-space distribution may be considered the analogue of the \textit{support distribution} (counting the number of codewords with a given support set). The terminology ``support'' is consistent with \cite{Ra18}.

In this paper, we show that that the row-space partition and the pivot partition are both reflexive, and that the row-space partition is also self-dual. We then compute the Krawtchouk coefficients of the row-space partition using a combinatorial approach based on M\"{o}bius inversion. This leads to an explicit MacWilliams identity for the row-space enumerator. We further introduce $U$-extremal codes (which generalize MRD codes) and show that they satisfy natural rigidity properties: (a) $U$-extremality is preserved by trace-duality; (b) for codes that are $U$-extremal for all~$U$ of a fixed dimension and below a fixed subspace, say~$T$, the partial row-space distribution below~$T$ only depends on the specified parameters, but not on the code or~$T$ itself (Theorem~\ref{T-Rigid}).

In the second part of the paper we study the dual of the pivot partition, showing that it can be naturally identified with the \textit{reverse pivot partition} $\cP^\rpiv$, where matrices are grouped according to their pivot indices after row reduction from the right. We then express the Krawtchouk coefficients of the partition pair $(\cP^\piv,\cP^\rpiv)$ in terms of the rank distribution of matrices supported on Ferrers diagrams (see Section~\ref{sec:KKpiv} for the precise definition of Ferrers diagram), establishing a MacWilliams identity in this context. We also provide both a recursive and an explicit formula for such rank distributions. Then we define pivot-extremal codes and show that they satisfy rigidity properties analogous to those for $U$-extremal codes.

Following work by Garsia/Remmel~\cite{GaRe86} and Haglund~\cite{Hag98}, in the third part of the paper we investigate connections between the rank distribution of matrices supported on Ferrers diagrams and
$q$-rook polynomials. 
The latter can be regarded as the $q$-analogues of classical rook polynomials associated with
a board; see \cite[Sections 7 and 8]{riordan2012introduction} for a general reference.
More precisely, as an application of our results, we give explicit expressions for the $q$-rook polynomials associated with a Ferrers board, and show that the number of matrices over $\F_q$ with rank~$r$ and
supported on a Ferrers diagram is a polynomial in $q$ whose degree strictly increases with $r$.

In the last part of the paper we characterize the linear maps on $\F_q^{n \times m}$ that preserve the rank, the row-space, or the pivot partition. We then give examples to show that in neither situation a MacWilliams Extension Theorem holds.

\paragraph*{Outline.} The paper is organized as follows.
In Section~\ref{sec:partitions} we recall the main definitions and results on partitions of finite abelian groups, Krawtchouk coefficients, rank-metric codes and MacWilliams identities. In Section~\ref{sec:ourpartitions} we introduce and establish the first properties of the row-space partition, the pivot partition and the reverse-pivot partition on the matrix space $\F_q^{n \times m}$. We devote Section~\ref{sec:KKrs} to the computation of the Krawtchouk coefficients of the row-space partition. In Section~\ref{S-UExtr} we define $U$-extremal codes and establish their rigidity properties. We compute the Krawtchouk coefficients of the pivot partition in Section~\ref{sec:KKpiv}, expressing them in terms of the rank distribution of matrices having a Ferrers diagram shape. Pivot-extremal codes are studied in Section~\ref{S-PivotExtr}. In Section~\ref{sec:FD} we give both a recursive and an explicit formula for the rank distribution of matrices supported on a Ferrers diagram. As a corollary, we show that the distribution is a polynomial in $q$. We then use these results to give a closed formula for the $q$-rook polynomials associated with Ferrers diagrams.
In Section~\ref{sec:pres} we study, for each of the three partitions, the partition-preserving linear maps,
and show that they do not satisfy the analogue of the MacWilliams Extension Theorem.

\section{Partitions and MacWilliams Identities}
\label{sec:partitions}
In this section we introduce partitions on matrix spaces and their character-theoretic dual.
We also define the Krawtchouk coefficients, which then determine the MacWilliams identities.

Throughout this paper, $q$ denotes a prime power and $\F=\F_q$ is the finite field with $q$ elements.
We denote by $\F^{n \times m}$ the space of $n \times m$ matrices over $\F$ and assume for the rest of the paper\footnote{In fact, the assumption is not needed for Sections~\ref{sec:KKpiv},~\ref{sec:FD}, and~\ref{sec:pres}.} that
\[
   m\leq n.
\]
Let $\N=\{1,2,3,\ldots\}$ and $\N_0=\{0,1,2,...\}$. For $i\geq 1$, we let $[i]:=\{1,...,i\}$.

Recall that the \textbf{trace product} of matrices $M,N \in\Fnm$ is
\begin{equation}\label{e-TraceProd}
    \langle A,B \rangle := \Tr(AB^\top),
\end{equation}
where $\Tr$ denotes the matrix trace.
Identifying $\Fnm$ with~$\F^{nm}$ via row concatenation, the trace product becomes the classical inner product of $\F^{nm}$.
Thus $(A,B) \longmapsto \langle A,B \rangle$ defines a symmetric and non-degenerate bilinear form on $\Fnm$.

\begin{defi}
Let $(G,+)$ be a group. The \textbf{character group} of $G$ is the set of all group homomorphisms $G \longrightarrow \C^*$ endowed with point-wise multiplication. It is denoted by $\widehat{G}$.
\end{defi}

It is well known (see for instance~\cite{Is76} for background on character theory) that if $G$ is a finite abelian group, then $G$ and $\widehat{G}$ are isomorphic (though not canonically so).
This is not the case for more general classes of groups.
Note also that if~$G$ is finite, then $|\chi(g)|=1$ for all~$g\in G$ and $\chi\in\widehat{G}$.
The character~$\chi$ given by $\chi(g)=1$ for all~$g\in G$ is called the \textbf{trivial character}.
For an $\F$-vector space $V$ we simply write $\widehat{V}$ for the character group of $(V,+)$.
Note that in this case $\widehat{V}$ carries a natural $\F$-vector space structure via
\begin{equation}\label{e-VhatVS}
   (c\chi)(v):=\chi(c v) \; \text{ for all }c\in\F,\;\chi\in\widehat{V},\; v\in V.
\end{equation}

Let $\chi:\F \longrightarrow \C^*$ be a non-trivial character of~$(\F,+)$.
The trace-product on $\Fnm$ induces via $\chi$ an
isomorphism of $\F$-vector spaces
\[
  \Fnm \longrightarrow \widehat{\Fnm},\qquad  B\longmapsto
  \left\{\begin{array}{ccc}\Fnm&\longrightarrow&\C^*\\  A&\longmapsto&\chi(\inner{A,B}).\end{array}\right.
\]
This isomorphism allows us to identify $\Fnm$ with its character group via the chosen character.
This identification is taken into account in the following definition.

\begin{defi}\label{D-DualPart}
Let $\cP=(P_i)_{i\in\cI}$ be a partition of $\Fnm$, and let $\chi$ be a non-trivial character of $\F$. The \textbf{dual}
of $\cP$ with respect to $\chi$ is the partition $\wcP$ of $\Fnm$ defined via the equivalence relation
\begin{equation}\label{e-Dual}
   B\widesim_{\wcP} \; B'\Longleftrightarrow
   \sum_{A\in P_i}\chi(\inner{A,B})=\sum_{A\in P_i}\chi(\inner{A,B'})\ \text{ for all }i\in\cI.
\end{equation}
We say that $\cP$ is \textbf{reflexive} if $\cP=\wwcP$  and
\textbf{self-dual} if $\cP=\wcP$. Note that self-duality implies reflexivity.
\end{defi}

One should be aware of the fact that the dual partition may depend on the choice
of the non-trivial character~$\chi$; see for example \cite[Ex.~2.2]{GL14homog}.
Therein, it is shown that even self-duality of a partition depends in general on the choice of the character.
Reflexivity, however, is independent of this choice.
This is a consequence of ~\cite[Prop.~4.4]{BGL15}.

For the partitions studied in this paper, the dual partitions do not depend on the choice of the character.
In fact, (as we will see) they belong to the following special class.

\begin{defi}
A partition $\cP=(P_i)_{i\in\cI}$ of $\Fnm$ is called \textbf{invariant} if $u P_i=P_i$ for all $u\in\F^*$ and $i\in\cI$,
that is, all blocks of the partition are invariant under multiplication by non-zero scalars.
\end{defi}

\begin{rem}\label{R-ChiInd}
Suppose $\cP=(P_i)_{i\in\cI}$ is an invariant partition of~$\Fnm$.
\begin{enumerate}
\item The dual partition~$\wcP$ does not depend on the choice of the non-trivial character~$\chi$.
        This follows from the fact that every other non-trivial character of~$\F$ is of the form $u\chi$ for some $u\in \F^*$; see~\eqref{e-VhatVS} for $u\chi$.
        Hence
        \begin{equation}\label{e-SumChiInv}
           \sum_{A\in P_i}(u\chi)(\inner{A,B})=\sum_{A\in P_i}\chi(u\inner{A,B})=\sum_{A\in P_i}\chi(\inner{uA,B})
           =\sum_{A\in P_i}\chi(\inner{A,B}),
        \end{equation}
        from which the statement follows.
\item The partition $\widehat{\mP}$ is invariant as well.
\end{enumerate}
\end{rem}

Now we are ready to introduce some fundamental parameters of invariant partitions.

\begin{defi}\label{D-Kraw}
Fix a non-trivial character~$\chi$ of~$\F$.
Let $\cP=(P_i)_{i\in\cI}$ be an invariant partition of $\Fnm$ and let $\wcP=(Q_j)_{j\in\cJ}$ be its dual partition.
For all $(i,j) \in \cI \times \cJ$, the complex number
\begin{equation}\label{e-Kraw}
  K(\mP;i,j):=\sum_{A\in P_i}\chi(\inner{A,B}), \quad \text{where }B\in Q_j ,
\end{equation}
is called the \textbf{Krawtchouk coefficient} of $\cP$ with index $(i,j)$.
Note that, thanks to~\eqref{e-SumChiInv}, the Krawtchouk coefficients do not depend on the choice of~$\chi$.
\end{defi}

We now introduce the main objects studied in this paper.

\begin{defi}\label{D-Code}
A (\textbf{matrix}) \textbf{code} is a linear subspace $\cC \le\Fnm$.
 The \textbf{dual} of $\mC$ is the matrix code
         \[
                \mC^\perp:=\{B \in\Fnm \mid \inner{A,B}=0 \mbox{ for all } A \in \mC\}.
         \]
         Observe that $\dim(\mC^\perp)=mn-\dim(\mC)$, and that $\mC^{\perp \perp}=\mC$.
\end{defi}

\begin{defi}\label{D-PartDistr}
Given a partition $\cP=(P_i)_{i\in\cI}$ of $\Fnm$ and a code $\cC \le \Fnm$, we define
$$\mP(\mC,i):= | \mC \cap P_i|, \quad i\in\cI.$$
We call the collection $(\mP(\mC,i))_{i \in \cI}$ the \textbf{$\cP$-distribution} of $\mC$.
\end{defi}

Now we can formulate a general version of the MacWilliams identities.
Such identities have been established various times for different settings:
for general subgroups of finite abelian groups in \cite[Thm.~4.72, Prop.~5.42]{Cam98} and~\cite[Thm.~2.7]{GL15Fourier},
for discrete subgroups of locally compact abelian groups in \cite[p.~94]{Fo98}, for codes over Frobenius rings
in \cite[Thm.~21]{HoLa01} and \cite[Thm.~2.11]{byrne2007linear}, and for codes supported on lattices in \cite[Thm.~29]{Ra18}.

\begin{theo}[MacWilliams Identities, see~\mbox{\cite[Thm.~2.7]{GL15Fourier}}] \label{th:MWI}
Let $\cQ=(Q_j)_{j\in\cJ}$ be an invariant partition of $\Fnm$ and let $\widehat{\cQ}=:\cP=(P_i)_{i\in\cI}$.
For all codes $\mC \le\Fnm$ and all $j\in\cJ$ we have
\[
    \cQ(\mC^\perp,j) =\frac{1}{|\cC|}\sum_{i\in\cI}K(\cQ;j,i) \: \cP(\cC,i).
\]
\end{theo}

Note that in the above formulation~$\cQ$ is the primal partition and~$\cP$ its dual.
The result tells us that the $\cQ$-distribution of~$\cC^{\perp}$ is fully determined by the $\cP$-distribution of~$\cC$.
The converse is not true in general.
However, if~$\cQ$ is reflexive, thus $\cQ=\widehat{\cP}$, then the two distributions mutually determine each other.

The MacWilliams identities give rise to the task to determine the Krawtchouk coefficients explicitly.
We will do so for various invariant partitions of~$\Fnm$, which we introduce in the next section.

\section{The Row-Space Partition and the Pivot Partition}
\label{sec:ourpartitions}
In this section we introduce the partitions mentioned in the title along with their character-theoretic duals.
Before doing so, we briefly discuss the rank partition. Recall that $m\leq n$.

\begin{defi}\label{rkP}
For $0 \le i \le m$ set
$P^\rk_i:=\{A\in\Fnm\mid \rk(A)=i\}$.
Then $\cP^\rk:=(P^\rk_r)_{0 \le r \le m}$ is a partition of $\Fnm$ of size $m+1$, called the \textbf{rank partition} of~$\Fnm$.
\end{defi}

This partition, which is clearly invariant, has been well studied in the past.
Self-duality is well-known but will also follow from our more general considerations later; see Corollary~\ref{C-SD}.
MacWilliams identities for additive codes endowed with the rank partition were first discovered by Delsarte~\cite[Thm.~3.3]{Del78} along with explicit expressions for the Krawtchouk coefficients
\cite[Thm.~A2]{Del78}; see also \cite[Ex.~39]{Ra18} for a proof using lattice theory.
They are given by
\begin{equation}\label{e-KrawRank}
   K(\cP^\rk;r,s)=\sum_{i=0}^m(-1)^{r-i}q^{ni+\BinomS{r-i}{2}}\GaussianD{m-i}{m-r}\GaussianD{m-s}{i} \quad \text{for all }0\leq r,s\leq m.
\end{equation}
Here $\Gaussian{a}{b}$ denotes the $q$-binomial coefficient.
It is the number of $b$-dimensional subspaces of~$\F_q^a$.

\medskip

We now turn to the partitions that will be the main subject of our investigation later on.
Let~$\cL$ be the set of all subspaces of~$\F^m$.
We have $\cL=\bigcup_{l=0}^m\cG_q(m,l)$, where $\cG_q(m,l)$ is the Grassmannian of $l$-dimensional subspaces of $\F^m$.
Then~$\cL$ is a lattice with respect to inclusion.

\begin{defi}\label{D-RSP}
For a matrix $A\in\Fnm$ we define $\rs(A):=\{uA\mid u\in\F^n\}$ to be the row space of~$A$.
For $U\in\cL$ set $P^\rs_U:=\{A\in\Fnm\mid \rs(A)=U\}$.
Then $\cP^\rs:=(P^\rs_U)_{U\in\cL}$ is a partition of $\Fnm$, called the \textbf{row-space partition} of~$\Fnm$.
\end{defi}

\begin{defi}\label{D-PivPart}
Define $\Pi=\{(j_1,\ldots,j_r)\mid 1\leq r\leq m,\, 1\leq j_1<\ldots<j_r\leq m\}\cup\{(\;)\}$, where $(\;)$ denotes the empty list.
For a list $\lambda\in\Pi$ we define $|\lambda|\in\{0,\ldots,m\}$ as its length.
For a matrix $A \in \Fnm$ we denote by $\RRE(A)$ the reduced row echelon form of~$A$, and define
\[
    \piv(A):=(j_1,\ldots,j_r)\in\Pi, \ \text{where }1\leq j_1<\ldots<j_r\leq m\text{ are the pivot indices of } \RRE(A).
\]
Then $\piv(0):=()$ and $|\piv(A)|=\rk(A)$ for all $A\in\Fnm$.
Matrices $A,B \in \Fnm$ are called \textbf{pivot-equivalent} if $\piv(A)=\piv(B)$.
This defines an equivalence relation on $\Fnm$. The equivalence classes form the \textbf{pivot partition} of $\Fnm$, denoted by $\cP^\piv$.
\end{defi}

Obviously,~$\Pi$ is in bijection to the set of all subsets of~$[m]$.
For us it will be helpful to record pivots as ordered lists, as introduced above.
We will use set-theoretical operations in the obvious way for pivot lists.

The three partitions defined above ($\cP^\rk$, $\cP^\rs$, and $\cP^\piv$) arise as the collection of orbits with respect to suitable group actions on $\mat$.
Indeed, consider the general linear groups of order~$n$ and~$m$ as well as the group
$\cU_m(\F)=\{S\in\GL_m(\F)\mid S\text{ is upper}$ $\text{triangular}\}$.
Define the actions
\begin{equation}\label{e-GroupActions}
\left.\begin{array}{ccccrcl}
  \rho_1:& \GL_n(\F)\times \Fnm                            &\longrightarrow&\Fnm,          & (S,A)&\longmapsto &\ SA,\\[.6ex]
  \rho_2:&(\GL_n(\F)\times\cU_m(\F))\times \Fnm  &\longrightarrow&\Fnm,         & (S,U,A)&\longmapsto &\ SAU^{-1},\\[.6ex]
  \rho_3:& (\GL_n(\F)\times\GL_m(\F))\times \Fnm&\longrightarrow&\Fnm, & (S,T,A)&\longmapsto &\ SAT^{-1}.
\end{array}\qquad \right\}
\end{equation}
Denote by $\cO_i$ the partition of~$\Fnm$ consisting of the orbits of~$\rho_i$.
We summarize some important properties of these partitions.

\begin{prop}\label{P-BasicsP}
\begin{enumerate}
\item \label{p1} $\cP^\rs \le \cP^\piv \le \cP^\rk$, that is, the row-space partition is finer than the pivot partition, which is finer than the rank partition.
\item \label{p2} $|\cP^\rk|=m+1$, $|\cP^\rs|=|\cL|=\sum_{l=0}^m\Gaussian{m}{l}$, and
$|\cP^\piv|=|\Pi|=\sum_{r=0}^m\Binom{m}{r}=2^m$.
\item \label{p3} $\cP^\rs=\cO_1$, $\cP^\piv=\cO_2$, and $\cP^\rk=\cO_3$.
\item \label{p4} $\cP^\rk$, $\cP^\rs$ and $\cP^\piv$ are invariant partitions.
\end{enumerate}
\end{prop}

\begin{proof} Property
(\ref{p1}) is clear and (\ref{p4}) is immediate from~(\ref{p3}).
Property (\ref{p2}) follows from the fact that for every possible rank $r\in\{0,\ldots,m\}$ we have $\Binom{m}{r}$
possibilities for the pivot indices of a matrix in $\Fnm$ with rank~$r$. The other two statements are clear.

Let us show (\ref{p3}). The identities concerning $\cP^\rs$ and $\cP^\rk$ are basic Linear Algebra.
It remains to show $\cP^\piv=\cO_2$.
Consider a matrix $A \in \Fnm$ and denote its columns by $A_1,\ldots,A_m$.
Then for any $j\in[m]$ we have
\begin{equation}\label{e-pivcol}
   j\in\piv(A)\Longleftrightarrow A_j\text{ is not in the span of the columns }A_1,\ldots,A_{j-1}.
\end{equation}
Let now $B=SAU^{-1}$ for some $S\in\GL_n(\F)$ and $U\in\cU_m(\F)$.
Then~\eqref{e-pivcol} immediately implies that $j\in\piv(A)\Longleftrightarrow j\in\piv(B)$ for any $j\in[m]$.
This proves $\cO_2\leq\cP^\piv$.
For the converse let $A,B\in\Fnm$ such that $\piv(A)=\piv(B):=(j_1,\ldots,j_r)$.
Let $\hat{A},\hat{B}$ be the RREF's of~$A,B$, respectively.
Then $\hat{A}=XA$ and $\hat{B}=YB$ for some $X,Y\in\GL_n(\F)$.
Denote by $e_1,\ldots,e_n$ the standard basis (column) vectors of ~$\F^n$.
Define the matrix $M=(M_1,\ldots,M_m)\in\Fnm$ via
\[
    M_i=\left\{\begin{array}{ll}e_\ell,&\text{if $i=j_\ell$ for some $\ell\in\{1,\ldots,r\}$},\\ 0,&\text{otherwise}\end{array}\right.
\]
In other words,~$M$ is obtained from~$\hat{A}$  (hence~$\hat{B}$) by keeping the pivot columns and erasing the others.
Now~\eqref{e-pivcol} implies
\[
   \hat{A}=MV,\quad \hat{B}=MW\text{ for some }V,W\in\cU_m(\F).
\]
Hence $B=Y^{-1}XAV^{-1}W$, and since $V^{-1}W$ is in $\cU_m(\F)$ we conclude that the matrices $A,B$ are in the same orbit of~$\cO_2$.
\end{proof}

We now turn to the duals of these partition.
The following more general result will be helpful.
It is a special case of \cite[Prop.~4.6]{BGL15}, where partitions induced by group actions are considered for arbitrary finite
Frobenius rings instead of finite fields.
For the sake of self-containment we provide a short proof.

\begin{prop}\label{P-OrbDual}
Let $\cS\leq\GL_n(\F)$ and $\cT\leq\GL_m(\F)$ be subgroups and define their transposes as $\cS'=\{S^\top\mid S\in\cS\}$ and $\cT'=\{T^\top\mid T\in\cT\}$.
Consider the group actions
\begin{alignat*}{7}
  \rho&:& \cS\times\cT\times \Fnm  &\longrightarrow\Fnm,  &\quad (S,T,A)\longmapsto &\ SAT^{-1},\\
  \rho'&:&\quad \cS'\times\cT'\times \Fnm  &\longrightarrow\Fnm,  &\quad (S,T,A)\longmapsto &\ SAT^{-1}.
\end{alignat*}
Let $\cO$ and $\cO'$ be the orbit partitions of~$\rho$ and~$\rho'$, respectively.
Then $\widehat{\cO}=\cO'$ and $\widehat{\cO'}=\cO$.
Thus, the partitions are reflexive and $|\cO|=|\cO'|$.
\end{prop}

\begin{proof}
We show that $\cO'\leq\widehat{\cO}$.
Let $B,B'\in\Fnm$ be in the same orbit of~$\cO'$, hence $B'=SBT$ for some $S\in\cS'$ and $T\in\cT'$.
For any orbit~$O$ of~$\cO$ we have $S^\top O T^\top=O$ and therefore
\begin{eqnarray*}
  \sum_{A\in O}\chi(\inner{A,B'}) &=& \sum_{A\in O}\chi(\inner{A,SBT}) = \sum_{A\in O}\chi(\tr(AT^\top B^\top S^\top)) = \sum_{A\in O}\chi(\tr(S^\top A T^\top B^\top)) \\ &=&
  \sum_{A\in O}\chi(\tr(AB^\top))= \sum_{A\in O}\chi(\inner{A,B}).
\end{eqnarray*}
Hence $\cO'\leq\widehat{\cO}$.
By symmetry we also have $\cO\leq\widehat{\cO'}$ and thus $\cO\leq\widehat{\phantom{\big|}\hspace*{.65em}}\hspace*{-1.1em}\widehat{\cO}$.
Since by \cite[Thm.~2.4]{GL15Fourier}  the converse is true for any partition, we conclude $\cO=\widehat{\phantom{\big|}\hspace*{.65em}}\hspace*{-1.1em}\widehat{\cO}$.
Furthermore, any partition~$\cP$ satisfies $|\cP|\leq|\wcP|$, see again \cite[Thm.~2.4]{GL15Fourier},
and thus we obtain $|\cO'|\leq |\widehat{\cO'}|\leq|\cO|\leq|\widehat{\cO}|$, where the middle step follows from $\cO\leq\widehat{\cO'}$.
Now the relation $\cO'\leq\widehat{\cO}$ implies $\cO'=\widehat{\cO}$.
The rest follows from symmetry.
\end{proof}

The following is now immediate with Proposition~\ref{P-BasicsP}(\ref{p3}).

\begin{cor}\label{C-SD}
$\cP^\rk=\widehat{\cP^\rk}$ and $\cP^\rs=\widehat{\cP^\rs}$, that is, the rank partition and the row-space partition are self-dual.
\end{cor}

In order to describe the dual of the pivot partition we need the \textit{reverse} pivot indices.
They are defined by performing Gaussian elimination on a matrix from right to left.
This is most conveniently defined using the matrix
\begin{equation}\label{e-SMat}
  Z=\begin{pmatrix} & & &1\\ & &1& \\ &\iddots& &\\ 1& & & \end{pmatrix}\in\GL_m(\F).
\end{equation}
Obviously, right multiplication of a matrix~$A$ by~$Z$ reverses the order of the columns of~$A$.

\begin{defi}\label{D-RPivPart}
Let $A \in \Fnm$ be a matrix and set $\hat{A}:=AZ$. Let $=\piv(\hat{A})=(j_1,\ldots,j_r)\in\Pi$.
Then we define the \textbf{reverse pivot indices} of~$A$ as
\[
   \rpiv(A)=(m+1-j_r,\ldots,m+1-j_1).
\]
We call $\RRE(\hat{A})Z$ the  \textbf{reverse reduced row echelon form} of~$A$.
Its pivot indices are $\rpiv(A)$.
Matrices $A,B \in \Fnm$ are called \textbf{reverse-pivot-equivalent} if $\rpiv(A)=\rpiv(B)$.
The resulting equivalence classes form the \textbf{reverse-pivot partition} of $\Fnm$, denoted by $\cP^\rpiv$.
\end{defi}

Note that $\rpiv(A)\in\Pi$, which means that the indices are ordered increasingly.
They satisfy the reverse analogue of~\eqref{e-pivcol}, i.e., for all $j \in [m]$
\begin{equation}\label{e-rpivcol}
   j\in\rpiv(A)\Longleftrightarrow A_j \text{ is not in the span of }A_{j+1},\ldots,A_m.
\end{equation}

In analogy to Proposition~\ref{P-BasicsP}(\ref{p3}), $\cP^\rpiv$ is the orbit partition of the group action~$\rho_2$ if we replace~$\cU_m(\F)$ by the
group of lower triangular invertible matrices.
Proposition~\ref{P-OrbDual} provides us with the following simple fact.

\begin{cor}\label{C-PpivDual} We have
$\widehat{\cP^\piv}=\cP^\rpiv$ and $\widehat{\cP^\rpiv}=\cP^\piv$. In particular, the partitions $\cP^\piv$ and $\cP^\rpiv$ are reflexive, but not self-dual.
\end{cor}

The above tells us that the pivot indices and the reverse pivot indices encode partitions that are mutually dual with respect to the trace inner product
$\inner{\,\cdot\,,\,\cdot\,}$ on $\Fnm$ as in~\eqref{e-TraceProd}.
In the remainder of this section we show how these indices reflect duality of subspaces in~$\F^m$ with respect to the standard inner product on~$\F^m$.
For $V\in\cL$ denote by~$V^\perp$ its orthogonal with respect to the standard inner product.
Furthermore, thanks to the uniqueness of the reduced row echelon form we may extend both the pivot partition and the reverse pivot partition to the lattice~$\cL$
of all subspaces of~$\F^m$: define $\piv(V)=\piv(A)$, where $A\in\F^{r\times m}$ is any matrix of full rank with row space~$V$, and define $\rpiv(V)$ similarly.
We need the following notion.

\begin{defi}\label{D-PivCompl}
Let $\lambda=(\lambda_1,\ldots,\lambda_r)\in\Pi$. We denote by $\widehat{\lambda}\in\Pi$ the \textbf{dual pivot list} of $\lambda$, that is,
$\widehat{\lambda}=(\hat{\lambda}_1,\ldots,\hat{\lambda}_{m-r})\in\Pi$ such that
$\{\lambda_1,\ldots,\lambda_r,\hat{\lambda}_1,\ldots,\hat{\lambda}_{m-r}\}=[m]$.
\end{defi}

Now we can show that for any subspace $V\in\cL$ the list of reverse pivot indices of the dual subspace~$V^{\perp}$ is the dual of the list of pivot indices of~$V$.
We will need this result later in Section~\ref{S-PivotExtr}.
Even though this is an entirely basic result from Linear Algebra, we were not able to find it in the literature and thus provide a proof.

\begin{prop}\label{P-VperpPiv}
Let $V\in\cL$ and $\piv(V)=\lambda$.
Then $\rpiv(V^\perp)=\widehat{\lambda}$.
\end{prop}

Let us first comment on this result.
Note that $\piv(V)$ may be regarded as an \textit{information set} (of minimal cardinality) of the code~$V$ in the classical sense (see \cite[p.~4]{HP03}).
More precisely, it is the lexicographically first one among all information sets of~$V$.
On the other hand, $\rpiv(V)$ is the first information set of~$V$ with respect to the reverse lexicographic order (that is, starting from the right).
Hence the above result tells us that the complement of the lexicographically first information set of~$V$ is the
reverse lexicographically first information set of~$V^\perp$.
In this sense, Proposition~\ref{P-VperpPiv} may be regarded as a refinement of \cite[Thm.~1.6.2]{HP03}.

\begin{proof}
Throughout this proof, for any matrix $M\in\F^{s\times m}$ we denote by $M_t$ the $t^{\rm th}$ column of~$M$.
Furthermore, we let $e_1,\ldots,e_m$ denote the standard basis vectors in~$\F^m$ and also use $e_1,\ldots,e_{m-r}$ as the standard basis vectors in~$\F^{m-r}$.
Let $\lambda=(\lambda_1,\ldots,\lambda_r)$ and $\widehat{\lambda}=(\hat{\lambda}_1,\ldots,\hat{\lambda}_{m-r})$.

Let $\dim(V)=r$ and let $A=(A_{ij})\in\F^{r\times m}$ be in RREF (reduced row echelon form) and such that $\rs(A)=V$.
Define the permutation matrix $P=(e_{\lambda_1},\ldots,e_{\lambda_r},e_{\hat{\lambda}_1},\ldots,e_{\hat{\lambda}_{m-r}})\in\GL_m(\F)$.
Then
\begin{equation}\label{e-bab}
   AP=(I_r\mid B),\text{ where $B=(B_{\alpha,\beta})=(A_{\alpha,\hat{\lambda}_\beta})\in\F^{r\times(m-r)}$ satisfies  $B_{\alpha,\beta}=0$ whenever $\hat{\lambda}_\beta<\lambda_\alpha.$}
\end{equation}
In other words, the pivot columns have been sorted to the front and the remaining columns appear in their original order in the matrix~$B$.
It follows that
\[
     V^\perp=\rs(M),\text{ where } M=((-B)^\top\mid I_{m-r})P^{-1}
\]
We show now that~$M$ is in reverse reduced row-echelon form with $\rpiv(M)=\widehat{\lambda}$ (see Definition~\ref{D-RPivPart}).

Condition~\eqref{e-bab} implies for the columns of $C:=(-B)^\top\in\F^{(m-r)\times r}$
\begin{equation}\label{e-CCol}
    C_{\alpha}\in\text{span}\{e_\beta\mid \hat{\lambda}_\beta> \lambda_\alpha\}.
\end{equation}
Hence the columns $M_t$ are given by
\[
   M_t=\left\{\begin{array}{cl} C_\alpha,&\text{if }t=\lambda_\alpha\text{ for some }\alpha=1,\ldots,r\\ e_\beta,&\text{if }t=\hat{\lambda}_\beta\text{ for some }\beta=1,\ldots,m-r.\end{array}\right.
\]
Thus~\eqref{e-CCol} reads as $M_{\lambda_\alpha}\in\text{span}\{M_{\hat{\lambda}_\beta}\mid\hat{\lambda}_\beta>\lambda_\alpha\}$, and this means
that~$\lambda_\alpha$ is not a reverse pivot index of~$M$; see~\eqref{e-rpivcol}.
As this is true for all $\alpha\in\{1,\ldots,r\}$ and $M$ has rank~$m-r$, we arrive at $\rpiv(V^\perp)=\rpiv(M)=(\hat{\lambda}_1,\ldots,\hat{\lambda}_{m-r})=\widehat{\lambda}$.
\end{proof}

\section{The Krawtchouk Coefficients of the Row-Space Partition}
\label{sec:KKrs}
In this section we explicitly determine the Krawtchouk coefficients of the row-space partition.
Recall that~$\cL$ denotes the lattice of all subspaces of~$\F^m$ and that $m\leq n$.

\begin{defi}\label{D-CodeEnum}
Let $\cC\leq\Fnm$ be a code. For $U\in\cL$ define
$\cC(U)=\{A\in\cC\mid \rs(A)\leq U\}$. Then $\cC(U)$ is a code as well (i.e., it is a linear subspace of $\cC$).
\end{defi}

Note that we consider two kinds of dual spaces:  the dual~$\cC^\perp$ of a matrix code $\cC\leq\Fnm$ with respect to the trace product (see Definition~\ref{D-Code}) and
the dual~$U^\perp$ of a subspace $U\in\cL$ with respect to the standard inner product on~$\F^m$.
These two kinds of dual spaces are related as follows.

\begin{lemma}[\mbox{\hspace*{-.4em}\cite[Lem.~28]{Ra16a}}]\label{P-CClCount}
Let $U\in\cL$ with $\dim U=u$. Then
\[
  \big|\cC(U)\big|=\frac{|\cC|}{q^{n(m-u)}}\big|\cC^\perp(U^\perp)\big|.
\]
\end{lemma}

Now we obtain the following explicit formulas for the Krawtchouk coefficients of~$\cP^\rs$.

\begin{theo} \label{T-KrawRS}
For all $U,V \in \cL$ we have
\[
    K(\cP^\rs;U,V)=\sum_{t=0}^{m} (-1)^{\dim(U)-t} \, q^{nt+ \binom{\dim(U)-t}{2}}\GaussianD{\dim(U \cap V^\perp)}{t}.
\]
\end{theo}

\begin{proof}
Fix a subspace $V\in\cL$ and let $M \in \Fnm$ be any matrix with $\rs(M)=V$.
Fix any non-trivial character $\chi$ of $\F$. Let
$f,g: \cL \longrightarrow \C$ be the functions defined, for all $U\in\cL$, by
\[
     f(U):=\sum_{\substack{ N \in \Fnm \\ \rs(N)=U}} \chi (\Tr(MN^\top)),\qquad
    g(U):=\sum_{U' \le U} f(U').
\]
Therefore $f(U)=K(\cP^\rs;U,V)$ for all $U \in \mL$; see Definition~\ref{D-Kraw}.
By Definition~\ref{D-CodeEnum} we have $\Fnm(U)=\{N\in\Fnm\mid \rs(N)\leq U\}$.
It follows that $\Fnm(U)^\perp=\Fnm(U^\perp)$ by \cite[Lem.~27]{Ra16a} and that $|\Fnm(U)| = q^{n\dim(U)}$ by Lemma~\ref{P-CClCount}.
Thus for all $U \in \mL$ we have
\[
g(U) = \!\!\!\sum_{\substack{ N \in \Fnm \\ \rs(N) \le U}}\!\!\!\!\! \chi (\Tr(MN^\top))
       = \!\!\!\sum_{N \in \Fnm(U)} \!\!\!\!\! \chi (\Tr(MN^\top))
         = \left\{ \begin{array}{cl}
         q^{n\dim(U)} & \mbox{ if } M \in \Fnm(U^\perp), \\ 0 & \mbox{ otherwise,}\end{array}  \right.
\]
where the last equality follows from the orthogonality relations of characters.
Denote by $\mu_\cL$ the M\"obius function of the lattice~$\cL$.
From~\cite[Ex.~3.10.2]{Sta97} we know
\begin{equation}\label{e-MoebL}
    \mu(W,V)=\left\{\begin{array}{cl}(-1)^{v-w}q^{\binom{v-w}{2}}&\text{if }W\leq V,\\ 0&\text{otherwise,}\end{array}\right.
\end{equation}
where $\dim W=w$ and $\dim V=v$.
Using that $M\in\Fnm(U'^\perp)$ iff $U'\leq V^\perp$, we thus obtain from M\"obius inversion
\[
    f(U)  = \sum_{U' \le U} g(U') \ \mu_\cL(U',U) = \sum_{U' \le U \cap V^\perp} q^{n\dim(U')} (-1)^{u-\dim U'}q^{\binom{u-\dim U'}{2}}
\]
 for all subspaces $U\in\cL$ with $\dim(U)=u$.
 As a consequence,
 \[
   f(U) = \sum_{t=0}^m \sum_{\substack{U' \le U \cap V^\perp \\ \dim(U')=t}} q^{nt}(-1)^{u-t} q^{\binom{u-t}{2}} =
      \sum_{t=0}^m  (-1)^{u-t} q^{nt+ \binom{u-t}{2}}\GaussianD{\dim(U \cap V^\perp)}{t}.
\]
 This gives the desired formula.
 \end{proof}

 Combining Theorem \ref{th:MWI} with Theorem \ref{T-KrawRS} one immediately obtains MacWilliams-type identities
 for the row-space partition.

 \begin{cor} \label{newMW}
 Let $\cC \le \Fnm$ be a code. Then for all $V\in\cL$ we have
 \[
      \cP^\rs(\cC^\perp,V) = \frac{1}{|\cC|}\sum_{U\in\cL} \cP^\rs(\cC,U)
        \sum_{t=0}^{m} (-1)^{\dim(V)-t} \, q^{nt+ \binom{\dim(V)-t}{2}}\GaussianD{\dim(V \cap U^\perp)}{t}.
\]
 \end{cor}

In the remainder of this section we provide different relations between the row-space partition distribution of a code~$\cC$ and that of~$\cC^\perp$.

\begin{prop}\label{C-WWperp}
Let $\cC\leq\Fnm$ be a matrix code.  Then for all $U\in\cL$ we have
\[
   \sum_{V\leq U}\cP^\rs(\cC,V)=\frac{|\cC|}{q^{n\dim U^\perp}}\sum_{W\leq U^\perp}\cP^\rs(\cC^\perp,W).
\]
\end{prop}

\begin{proof}
Using Lemma \ref{P-CClCount} we obtain
\[
 \sum_{V\leq U}\cP^\rs(\cC,V)=|\cC(U)|=\frac{|\cC|}{q^{n\dim U^\perp}}|\cC^\perp(U^\perp)|=\frac{|\cC|}{q^{n\dim U^\perp}}\sum_{W\leq U^\perp}\cP^\rs(\cC^\perp,W).
 \qedhere
\]
\end{proof}

The last proposition gives $N$ linear relations, where $N=|\cL|$.
They may be written as a linear system as follows.
Define the row vectors
\[
    \cP^\rs(\cC)=\big(\cP^\rs(\cC,V)\big)_{V\in\cL},\quad \cP^\rs(\cC^\perp)=\big(\cP^\rs(\cC^\perp,V)\big)_{V\in\cL}\in\C^N,
\]
which describe the partition distribution of the codes $\cC$ and $\cC^\perp$, respectively; see Definition~\ref{D-PartDistr}.
Then Proposition~\ref{C-WWperp} reads as
\[
   \cP^\rs(\cC) \cdot A=|\cC|\cdot \cP^\rs(\cC^\perp) \cdot B\cdot D,
\]
where $A,B,D\in\C^{N\times N}$ are defined as
\[
  A(V,U):=\left\{\begin{array}{cc}1&\text{if }V\leq U \\ 0 &\text{otherwise}\end{array}\right.,\qquad
  B(V,U):=A(V,U^\perp), \quad D:=\text{diag}{\left(1/q^{n\dim(U^\perp)}\right)}_{U\in\cL}.
\]
The matrix~$A$ may be regarded as the $\zeta$-function of the subspace lattice~$\cL$.
Thus its inverse is the M{\"o}bius function, which shows that~$A$ is invertible.
The same is true for the matrix~$B$. Therefore we have
\begin{equation}\label{e-MacWId}
   \cP^\rs(\cC^\perp)=\frac{1}{|\cC|} \cdot \cP^\rs(\cC) \cdot M, \qquad \mbox{where } M:=A \cdot \text{diag}{\left(q^{n\dim(U^\perp)}\right)}_{U\in\cL} \cdot B^{-1}.
\end{equation}
This provides us with a different method to compute the enumerators $\cP^\rs(\cC^\perp,U)$ from the enumerators $\cP^\rs(\cC,V)$ for $V\in\cL$.
The entries of the matrix $M\in\C^{N\times N}$ are the Krawtchouk coefficients of the row-space partition~$\cP^\rs$.
This follows, for instance, from \cite[Thm.~2.7]{GL15Fourier}.

We close this section by presenting  the binomial moments of the row-space distribution.
They consist of $m+1$ identities and form the analogue to those for the Hamming weight in~$\F^n$ (see \cite[(M$_2$) on p.~257]{HP03}) and for the rank weight
(see \cite[Prop.~4]{GaYa08} for $\F_{q^m}$-linear rank-metric codes and~\cite[Thm.~31]{Ra16a} for $\F_q$-linear rank-metric codes).

\begin{prop}\label{P-WVCount}
Let $\cC\leq\Fnm$ be a matrix code. Then for all integers $0\leq\nu\leq m$ we have
\[
  \sum_{V\in\cL}\GaussianD{m-\dim V}{\nu}\cP^\rs(\cC,V)=\frac{|\cC|}{q^{n\nu}}\sum_{W\in\cL}\GaussianD{m-\dim W}{m-\nu}\cP^\rs(\cC^\perp,W).
\]
\end{prop}

\begin{proof}
By \cite[Eq.~(8)]{GoRa17a}, for all $\cC \le \Fnm$ and all $0 \le \nu \le m$ we have
\[
      \sum_{\substack{ U\in\cL\\ \dim U=m-\nu}}|\cC(U)|=\sum_{i=0}^{m-\nu}\displaystyle \GaussianD{m-i}{\nu} \cP^\rk(\cC,i).
\]
Therefore
\begin{align*}
  \sum_{V\in\cL}\GaussianD{m-\dim V}{\nu}\cP^\rs(\cC,V)&=\sum_{i=0}^m\GaussianD{m-i}{\nu} \sum_{\substack{V\in\cL \\ \dim(V)=i}} \cP^\rs(\cC,V)
       =\sum_{i=0}^m\GaussianD{m-i}{\nu} \cP^\rk(\cC,i) \\
     &=\sum_{\substack{U\in\cL \\ \dim U=m-\nu}}|\cC(U)|.
\end{align*}
Similarly,
\[
\sum_{W\in\cL}\GaussianD{m-\dim W}{m-\nu}\cP^\rs(\cC^\perp,W) = \sum_{\substack{U\in\cL\\ \dim U=\nu}}|\cC^\perp(U)|.
\]
Using Lemma~\ref{P-CClCount}  we obtain
\begin{eqnarray*}
  \sum_{V\in\cL}\GaussianD{m-\dim V}{\nu} \cP^\rs(\cC,V) &=& \frac{|\cC|}{q^{n\nu}}\sum_{\substack{U\in\cL \\ \dim U=m-\nu}}|\cC^\perp(U^\perp)| \\
 &=& \frac{|\cC|}{q^{n\nu}} \sum_{\substack{ U\in\cL \\ \dim U=\nu}}|\cC^\perp(U)|\\
  &=& \frac{|\cC|}{q^{n\nu}}\sum_{W\in\cL}\GaussianD{m-\dim W}{m-\nu} \cP^\rs(\cC^\perp,W),
\end{eqnarray*}
for all $0 \le \nu \le m$, which is the desired equation.
\end{proof}

\section{$U$-Extremal Codes}
\label{S-UExtr}
In this section we generalize the notion of MRD codes to matrix codes with  respect to the row-space partition.
Let us first recall the following facts.
As before, we assume that $m\leq n$.

Recall that $(\Fnm,\rho)$, where $\rho(A,B)=\rk(A-B)$, is a metric space.
A subspace $\cC$ of this metric space is called a \textbf{rank-metric code}.
Its  (\textbf{minimum rank}) \textbf{distance}  is defined as
\begin{equation}\label{e-drk}
    \drk(\cC):=\min\{\rk(A-B)\mid A,B\in\cC,\,A\neq B\}=\min\{\rk(M)\mid M\in\cC\setminus\{0\}\}.
\end{equation}
The Singleton-like bound for rank-metric codes~\cite[Thm.~5.4]{Del78} tells us that if $\cC\leq\Fnm$ is a non-zero code of distance~$d$, then
\begin{equation}\label{e-Singl}
   |\cC|\leq q^{n(m-d+1)}.
\end{equation}
A code~$\cC\leq\Fnm$ is an \textbf{MRD code} if $\cC=\{0\}$ or  if
$\cC \neq \{0\}$ and $|\cC|= q^{n(m-d+1)}$, where $d=\drk(\cC)$.
In other words, MRD codes are \textit{extremal} with respect to the Singleton-like bound.

MRD codes enjoy various properties.
We briefly list these properties and then generalize the concept to matrix codes with  respect to the row-space partition.

\begin{rem}\label{R-MRDExtr}
Let $\cC\leq\Fnm$ an MRD code. The following hold.
\begin{enumerate}
\item The dual code $\mC^\perp$ is MRD as well. Moreover, if $\mC \neq \{0\}$ has minimum distance $d$, then $\cC^\perp$ has minimum distance $m-d+2$; see~\cite[Thm.~5.5]{Del78} or also~\cite[Cor.~41]{Ra16a}.
\item The rank distribution of $\cC$ only depends only on the parameters
$q,n,m,d$ of the code; see~\cite[Thm.~5.6]{Del78} or \cite[Cor.~44]{Ra16a}.
\label{r2}
\end{enumerate}
\end{rem}

MRD codes are even more rigid than stated in~(\ref{r2}) above: even their row-space distribution depends only on their parameters, as the following result shows.
This fact also follows from the proof of \cite[Thm. 8]{GoRa17a}. Recall the notation~$\cL$ for the lattice of subspaces in~$\F^m$ as well as the notation for partition distributions in Definition~\ref{D-PartDistr}.

\begin{theo}\label{P-MRDRSDistr}
Let $\cC\leq\Fnm$ be a non-zero MRD code of minimum distance $d$, and let $V\in\cL$ with $\dim(V)=v$. Then
        \[
          \cP^\rs(\cC,V)=\sum_{i=0}^{d-1}\GaussianD{v}{i}(-1)^{v-i}q^{\BinomS{v-i}{2}}+\sum_{i=d}^{v}\GaussianD{v}{i}q^{n(i-d+1)}(-1)^{v-i}q^{\BinomS{v-i}{2}}.
        \]
In particular, the row-space distribution of $\cC$ depends only on the parameters $q,n,m,d$.
\end{theo}

One may note that the above expression actually does not explicitly depend on~$m$. This parameter only enters via the lattice~$\cL$.

\begin{proof}
Fix $V\in\cL$ with $\dim(V)=v$. It follows from \cite[Lem.~48]{Ra18} (see also \cite[Lem.~25]{GoRa17a}) that
\begin{equation}\label{e-CU}
   |\cC(V)|=\left\{\begin{array}{cl} 1 &\text{if }0\leq v\leq d-1,\\ q^{n(v-d+1)}&\text{if }v\geq d.\end{array}\right.
\end{equation}
 Define functions $f,\,g:\cL \longrightarrow \R$ by
\[
  f(V)=\cP^\rs(\cC,V) \quad \text{ and } \quad g(V)=\sum_{U\leq V} f(U)
\]
for all $V \in \mL$. Then $g(V)=|\cC(V)|$ by definition. Using M\"obius inversion in the lattice $\mL$ and~\eqref{e-MoebL} we compute
\begin{align*}
  f(V)&=\sum_{U\leq V}g(U)\mu_{\cL}(U,V)\\
       &=\sum_{\substack{U\leq V\\ \dim(U)\leq d-1}}\hspace*{-1.2em}(-1)^{v-\dim(U)}q^{\BinomS{v-\dim(U)}{2}}+
          \hspace*{-.7em}\sum_{\substack{U\leq V\\ \dim(U)\geq d}}\hspace*{-1.2em}q^{n(\dim(U)-d+1)}(-1)^{v-\dim(U)}q^{\BinomS{v-\dim(U)}{2}}.
\end{align*}
The desired identity follows from the fact that~$V$ contains $\Gaussian{v}{i}$ subspaces of dimension~$i$.
\end{proof}

Note that using the $q$-binomial theorem \cite[p. 74]{Sta97}  one easily confirms that $\cP^\rs(\cC,V)=0$ whenever $\dim(V)\in\{1,\ldots,d-1\}$.

\medskip
We now propose a generalization of the Singleton-type bound for matrix codes.
This will lead to a refined notion of extremality.

\begin{prop} \label{bound}
Let $\cC \le \Fnm$ and $U \in \mL$ with $u:=\dim(U)$. Assume $\cC(U)=\{0\}$. Then we have
$|\cC| \le q^{n(m-u)}$.
\end{prop}

\begin{proof}
From \cite[Lem.~26]{Ra16a} we know that $\dim(\Fnm(U))=nu$. Therefore $0= \dim(\cC(U))=\dim(\cC \cap \Fnm(U)) \ge \dim(\cC) + nu - nm$,
which results in the stated bound.
\end{proof}

The following generalization of MRD codes is natural from the previous result.
We will see that these codes satisfy similar rigidity properties as listed for MRD codes in Remark~\ref{R-MRDExtr}.

\begin{defi} \label{def:Uex}
Let $\cC \le \Fnm$ and $U \in \mL$ with $u:=\dim(U)$. We say that
$\cC$  is \textbf{$U$-extremal} if $\cC(U)=\{0\}$ and $|\cC|=q^{n(m-u)}$.
\end{defi}

Clearly, $\Fnm$ is the only $\{0\}$-extremal code and, dually, $\{0\}$ is the only $\F^m$-extremal code.
The Singleton-like bound~\eqref{e-Singl} implies that if~$\cC$ is $U$-extremal, then $\dim(U)\geq\drk(\cC)-1$.
This immediately leads to the following observation.

\begin{rem}\label{R-MRDExtr2}
Let $\cC \le \Fnm$ be a non-zero code of minimum distance $d$. The following are equivalent.
\begin{enumerate}
\item $\cC$ is an MRD code,
\item $\cC$ is $U$-extremal for all $U\in\cL$ with $\dim(U)=d-1$,
\item $\cC$ is $U$-extremal for some $U\in\cL$ with $\dim(U)=d-1$.
\end{enumerate}
\end{rem}

There exist $U$-extremal codes that are not MRD.

\begin{exa}\label{E-NotMRDExtr}
Write $m=m_1+m_2$ with $m_1,m_2 \neq 0$. Let $\cC_1 \le \F^{n\times m_1}$ be a non-zero MRD code of minimum distance $d$, say.
Define $\cC=\{(A\mid 0)\in\F^{n\times(m_1+m_2)}\mid A\in \cC_1\}$.
Then $\cC\le\F^{n\times(m_1+m_2)}$ has cardinality $q^{n(m_1-d+1)}$ and minimum distance~$d$, thus~$\cC$ is not MRD.

Choose any subspace $U_1 \leq\F^{m_1}$ of dimension $d-1$ and set $U=U_1\times\F^{m_2}$.
Then $\dim(U)=m_2+d-1$ and thus $|\cC|=q^{n(m_1+m_2-\dim(U))}$.
In order to see that the code~$\cC$ is~$U$-extremal, let $(A\mid 0)\in\cC$ such that $\rs(A\mid 0)\leq U$.
Then $\rs(A)\leq U_1$ and thus $\rk(A)\leq d-1$.
But then $A=0$, and all of this shows that $\cC(U)=\{0\}$.
Hence $\cC$ is $U$-extremal, but not MRD.
\end{exa}

Extremality is preserved under dualization. The following result is an immediate consequence of the definitions and
Lemma \ref{P-CClCount}.

\begin{prop}\label{dualis}
Let $\cC \le \Fnm$ and $U\in\cL$. 
Then~$\cC$ is $U$-extremal if and only if~$\cC^\perp$ is $U^\perp$-extremal.
\end{prop}

For $\F_{q^m}$-linear codes, $U$-extremality is related to \textit{information spaces}, defined in~\cite[Sec.~VI]{MP16}.
These are spaces such that the according puncturing map (in a suitable sense) is injective on the code, see \cite[Def.~12]{MP16}, and hence preserves all information.
Restricting our considerations to $\F_{q^m}$-linear codes, one easily observes that $\cC(U)=\{0\}$ iff $U^{\perp}$ is an information space of minimal dimension.
Hence Proposition~\ref{dualis} states that~$U^{\perp}$ is a minimal information space of~$\cC$ iff~$U$ is a minimal information space of~$\cC^{\perp}$.
This is a special case of \cite[Prop.~15]{MP16}. 
All of this tells us that the above result may be regarded as an analogue of \cite[Thm.~1.6.2]{HP03} for matrix codes.


Next we turn to the row-space distribution of $U$-extremal codes.
It cannot be expected that the entire distribution depends only on the parameters of the code and the dimension of~$U$.
The following example illustrates this.

\begin{exa}\label{E-NotMRDRSD}
Consider again Example~\ref{E-NotMRDExtr}, taking $m_2=1$.
Then we have $\dim(U)=d$ and, of course, $\cP^\rs(\cC,U)=0$.
On the other hand, choose any $d$-dimensional subspace $\tilde{V}\leq\F^{m_1}$ and set $V=\tilde{V}\times\{0\}$.
Then $\cP^\rs(\cC,V)=\cP^\rs(\tilde{\cC},\tilde{V})=|\tilde{\cC}(\tilde{V})|-1=q^n-1$ thanks to~\eqref{e-CU}.
Thus $\cP^\rs(\cC,V)\neq\cP^\rs(\cC,U)$ even though $\dim(V)=\dim(U)$.
\end{exa}

However, we do obtain a rigidity result in the case where $\cC$ is $U$-extremal for all subspaces $U$ of a fixed dimension contained in a given $T\in\cL$.
The following result generalizes Theorem~\ref{P-MRDRSDistr}.

\begin{theo}[Rigidity of extremality]\label{T-Rigid}
Let $\cC \le \Fnm$ and $T \in \mL$. Let $0 \le u \le \dim(T)$ be an integer, and suppose that $\cC$ is $U$-extremal for all $U \le T$ of dimension $u$. Then for all $V \in \mL$ with $V \le T$ we have
        \begin{equation}\label{e-RSRigid}
          \cP^\rs(\cC,V)=\sum_{i=0}^{u}\GaussianD{v}{i}(-1)^{v-i}q^{\BinomS{v-i}{2}}+\sum_{i=u+1}^{v}\GaussianD{v}{i}q^{n(i-u)}(-1)^{v-i}q^{\BinomS{v-i}{2}},
        \end{equation}
where $v=\dim(V)$. Hence the partial row-space distribution $(\cP^\rs(\cC,V))_{V\leq T}$ depends only on $n,q,u$.
\end{theo}

Note the extreme case where $u=\dim(T)$, in which the assumptions simply mean that~$\cC$ is $T$-extremal.
This clearly implies $\cP^\rs(\cC,V)=0$ for all $0<V\leq T$, which also follows from~\eqref{e-RSRigid} along with the $q$-binomial formula.
More interestingly, we also recover Theorem~\ref{P-MRDRSDistr}: choose $T=\F^m$ and $u=d-1$.
Then the above assumption means that~$\cC$ is MRD, see Remark~\ref{R-MRDExtr2}, and~\eqref{e-RSRigid}
coincides with Theorem~\ref{P-MRDRSDistr}.

\begin{proof}
Let $V \le T$ have dimension $v$. We show first that
\begin{equation}\label{e-CV}
    |\cC(V)|=\left\{\begin{array}{cl} 1 &\text{if }0\leq v\leq u,\\ q^{n(v-u)}&\text{if }v>u.\end{array}\right.
\end{equation}
Indeed, if $v \le u$ then there exists $U \in \mL$ such that $\dim(U)=u$ and  $V \le U \le T$.
Since $\mC$ is $U$-extremal, we have $\mC(V) \le \mC(U)=\{0\}$. Therefore $|\cC(V)|=1$.
Now suppose that $v>u$, and fix a $u$-dimensional space $U \in \mL$ with $U \le V$.
By Proposition~\ref{dualis}, $\cC^\perp$ is $U^\perp$-extremal. Therefore
$\cC^\perp(V^\perp) \le \cC^\perp(U^\perp) =\{0\}$.
Thus by Lemma~\ref{P-CClCount} we conclude $|\cC(V)|=\frac{|\cC|}{q^{n(m-v)}}=q^{n(v-u)}$.
This establishes~\eqref{e-CV}.

We can now proceed as in the proof of Theorem \ref{P-MRDRSDistr}, this time using M\"obius inversion on
the interval $[0,T]$ of $\mL$.
\end{proof}

We conclude this section by observing that there are indeed non-MRD codes that satisfy the assumptions of Theorem \ref{T-Rigid}.
The example also shows that the result just proven does not extend to subspaces that are not contained in~$T$.

\begin{exa}\label{E-RSRigidNotMRD}
\begin{enumerate}
\item Let $n,m_1,m_2 \ge 1$ be integers with $n \ge m_1+m_2 \ge u+1 \ge 2$. Set $m:=m_1+m_2$, and let $\mC_1 \le \F^{n \times m_1}$ be an MRD code with minimum distance $u+1$. Then
$\mC_1$ has dimension $n(m_1-u)$. Moreover, $\mC_1(U)=\{0\}$ for all $U \le \F^{m_1}$ of dimension $u$. Construct the code
\[
    \mC:=\left\{(A \mid B) \in \F^{n \times m} \mid A \in \mC_1, \; B \in \F^{n \times m_2}\right\}.
\]
Let $T:=\F^{m_1} \times 0^{m_2} \le \F^m$.
Clearly, $\mC(U)=\{0\}$ for all $U \le T$ of dimension $u$. Moreover, $\dim(\mC)=n(m-u)$.
Thus $\mC$ is $U$-extremal for all $U \le T$ of dimension $u$. Note that $\mC$ is not MRD, as its rank distance is 1 and $\dim(\mC)=n(m-u)<nm$.

\item Consider the code from (1). By Theorem \ref{T-Rigid} we know that $\cP^\rs(\cC,V)$ only depends $\dim(V)$ for all
$V \le T$. Note however that this is not the case in general for the spaces $V$ that are not contained in $T$. Let e.g. $V_1= \langle e_1,...,e_u,e_m\rangle$ and
$V_2= \langle e_m,e_{m-1},...,e_{m-u}\rangle$, where $\{e_1,...,e_m\}$ is the canonical basis of $\F^m$.
The spaces $V_1$ and $V_2$ have the same dimension, $u+1$, and neither of them is contained in $T$. Suppose $m_2 \ge u+1$. Then
$\cP^\rs(\cC,V_1)=0$ and $\cP^\rs(\cC,V_2)=\prod_{i=0}^{u} (q^n-q^i)$.
\end{enumerate}
\end{exa}

\section{The Krawtchouk Coefficients of the Pivot Partition}
\label{sec:KKpiv}
This section is devoted to obtaining explicit formulas for the Krawtchouk coefficients of the pivot partition, introduced in Definition~\ref{D-PivPart}.
They will be expressed in terms of the rank distribution of matrices that are supported on a Ferrers diagram. We therefore start by introducing the needed notation and terminology. In this section we do \textit{not} assume $m \le n$.

\begin{defi}\label{D-Ferrers}
An $n\times m$ \textbf{Ferrers diagram} (or \textbf{Ferrers board}) $\cF$ is a subset of $[n]\times[m]$ that satisfies the following properties:
\begin{enumerate}
\item if $(i,j)\in\cF$ and $j<m$, then $(i,j+1)\in\cF$ (right aligned),
\item if $(i,j)\in\cF$ and $i>1$, then $(i-1,j)\in\cF$ (top aligned).
\end{enumerate}
For $j=1,\ldots,m$ let  $c_j=|\{(i,j)\mid (i,j)\in\cF,\,1\leq i\leq n\}|$.
Then we may identify the Ferrers diagram~$\cF$ with the tuple $[c_1,\ldots,c_m]$.
It satisfies $0\leq c_1\leq c_2\leq\ldots\leq c_m\leq n$.
\end{defi}

The Ferrers diagram $\cF=[c_1,\ldots,c_m]$ can be visualized as an array of top-aligned and right-aligned dots
where the $j$-th column has $c_j$ dots.
Just like for matrices, we index the rows from top to bottom and the columns from left to right.
For instance, $\cF=[1,2,4,4,5]$ is given by
     
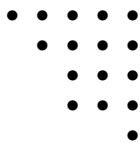
\begin{figure}[ht]
    \centering
     {\small
     \begin{tikzpicture}[scale=0.4]
         \draw (4.5,1.5) node (b1) [label=center:$\bullet$] {};
         \draw (4.5,2.5) node (b1) [label=center:$\bullet$] {};
         \draw (4.5,3.5) node (b1) [label=center:$\bullet$] {};
         \draw (4.5,4.5) node (b1) [label=center:$\bullet$] {};
         \draw (4.5,5.5) node (b1) [label=center:$\bullet$] {};
       \

         \draw (3.5,2.5) node (b1) [label=center:$\bullet$] {};
         \draw (3.5,3.5) node (b1) [label=center:$\bullet$] {};
         \draw (3.5,4.5) node (b1) [label=center:$\bullet$] {};
         \draw (3.5,5.5) node (b1) [label=center:$\bullet$] {};

         \draw (2.5,2.5) node (b1) [label=center:$\bullet$] {};
         \draw (2.5,3.5) node (b1) [label=center:$\bullet$] {};
         \draw (2.5,4.5) node (b1) [label=center:$\bullet$] {};
         \draw (2.5,5.5) node (b1) [label=center:$\bullet$] {};

         \draw (1.5,4.5) node (b1) [label=center:$\bullet$] {};
        \draw (1.5,5.5) node (b1) [label=center:$\bullet$] {};

       \draw (0.5,5.5) node (b1) [label=center:$\bullet$] {};
     \end{tikzpicture}
     }
    \end{figure}

We expressly allow $c_1=0$ or $c_m<n$.
This has the consequence that for all $\tilde{n}\leq n$ and $\tilde{m}\leq m$ an $\tilde{n}\times\tilde{m}$ Ferrers diagram is also an $n\times m$ Ferrers diagram.
Moreover, the empty Ferrers diagram is given by $\cF=[0,\ldots,0]$ of any length.

\begin{defi}\label{D-FMat}
The (\textbf{Hamming}) \textbf{support} of a matrix $M=(M_{ij})\in\Fnm$ is defined as the index set $\supp(M):=\{(i,j)\mid M_{ij}\neq0\}$.
The subspace of~$\Fnm$ of all matrices with support contained in~$\cF$ is denoted by $\F[\cF]$.
For $r \in \{0,...,m\}$ we set $P_r(\cF)=\cP^\rk(\F[\cF],r)$, that is,
\[
    P_r(\cF)=|\{M\in\F[\cF]\mid \rk(M)=r\}|.
\]
We call $(P_r(\cF))_{0 \le r \le m}$ the \textbf{rank-weight distribution} of $\F[\cF]$.
Clearly, $P_0(\cF)=1$ for any Ferrers diagram~$\cF$, including the empty one.
\end{defi}

The following result provides an explicit formula for the rank-weight distribution of the space~$\F[\cF]$ for any Ferrers diagram~$\cF$.
We postpone the proof to Section~\ref{sec:FD}, where we will describe connections between the rank-weight distribution of $\F[\cF]$ and $q$-rook polynomials.
For all $r\in\N$ define
\begin{equation}\label{e-Irm}
  \cI_{r,m}:=\{(i_1,\ldots,i_r)\mid 1\leq i_1<\ldots<i_r\leq m\}.
\end{equation}
Clearly $\cI_{r,m}=\emptyset$ if $r>m$. Moreover, for $i=(i_1,\ldots,i_r)\in\cI_{r,m}$ set
\begin{equation}\label{e-len}
   \len{i}:=\sum_{j=1}^r i_j.
\end{equation}
It will be convenient to set $\cI_{0,m}=\{()\}$ and $\len{()}=0$.

\begin{theo} \label{th:new}
Let $\cF=[c_1,\ldots,c_m]$ be an $n\times m$ Ferrers diagram, and let
$r\in\N_0$.
Then
\begin{equation}\label{e-RkW}
   P_r(\cF)=\sum_{i=(i_1,\ldots,i_r)\in\cI_{r,m}}q^{rm-\len{i}}\prod_{j=1}^r(q^{c_{i_j}-j+1}-1).
\end{equation}
\end{theo}

We also need the following technical result.

\begin{lemma}\label{L-AB}
Let $\sigma=(\sigma_1,\ldots,\sigma_b)\in\Pi$ and $B\in\F^{b\times m}$ be the matrix with columns
\[
     B_j=\left\{\begin{array}{cl} e_{\alpha} &\text{if }j=\sigma_{\alpha},\\ 0 &\text{else}.\end{array}\right.
\]
Thus~$B$ is in RREF with $\piv(B)=\sigma$, and where all non-pivot columns are zero.
Let $\lambda=(\lambda_1,\ldots,\lambda_a)\in\Pi$ and set
\[
    \lambda\cap\sigma=(\lambda_{\alpha_1},\ldots,\lambda_{\alpha_x})\ \text{ and }\widehat{\sigma}\setminus\lambda=(\hat{\sigma}_{\beta_1},\ldots,\hat{\sigma}_{\beta_y}).
\]
Furthermore, for $j\in[y]$ set $z_j=|\{i\mid \lambda_{\alpha_i}<\hat{\sigma}_{\beta_j}\}|$.
Then for any $r\in\{0,\ldots,a\}$ we have
\[
  \Big|\{A\in\F^{a\times m}\mid A\text{ is in RREF},\, \piv(A)=\lambda,\, \rk\begin{pmatrix}A\\B\end{pmatrix}=b+r\Big\}\Big|=P_{r-a+x}(\cF),
\]
where $\cF$ is the $x\times y$ Ferrers diagram $\cF=[z_1,\ldots,z_y]$ and $P_t(\cF)$ is the rank-weight distribution of $\F[\cF]$ from Theorem~\ref{th:new}.
\end{lemma}

From $\hat{\sigma}_{\beta_1}<\ldots<\hat{\sigma}_{\beta_y}$ we conclude $z_1<\ldots<z_y\leq x$.
Hence~$\cF$ is indeed an $x\times y$ Ferrers diagram.
We may have $z_1=0$ and the Ferrers diagram could be shortened by removing  empty columns.
Precisely, let $t'$ be minimal such that $\hat{\sigma}_{\beta_{t'}}> \lambda_{\alpha_1}$.
Then $z_{t'}\neq0=z_{t'-1}$.
Note also that for the given matrix~$B$ and any matrix~$A$ as specified above we have
$\rk\SmallMat{B}{A}\geq b+|\lambda\setminus\sigma|=b+(a-x)$.
Hence only $r\geq a-x$ matters in the above formula.
Before giving the proof of Lemma~\ref{L-AB}, we illustrate the count by an example.

\begin{exa}\label{E-AB}
Let $m=7$ and $\sigma=(3,4,6),\,\lambda=(1,4,6)$. Then $\widehat{\sigma}=(1,2,5,7)$ and
\[
  \lambda\cap\sigma=(4,6)=(\lambda_2,\lambda_3) \ \text{ and }\
  \widehat{\sigma}\setminus\lambda=(2,5,7).
\]
Using~$*$ for the unspecified entries of the matrix~$A$ in RREF we observe
\[
  \rk\begin{pmatrix}B\\\hline A\end{pmatrix}
  =\rk\begin{pmatrix}0&0&1&0&0&0&0\\0&0&0&1&0&0&0\\0&0&0&0&0&1&0\\\hline1&*&*&0&*&0&*\\0&0&0&1&*&0&*\\0&0&0&0&0&1&*\end{pmatrix}
  =\rk\begin{pmatrix}0&0&1&0&0&0&0\\0&0&0&1&0&0&0\\0&0&0&0&0&1&0\\\hline1&*&0&0&*&0&*\\0&0&0&0&*&0&*\\0&0&0&0&0&0&*\end{pmatrix},
\]
where we applied row operations to clear the columns of~$A$ using the pivot positions of~$B$.
Clearing the rows of~$A$ that still contain pivots shows that
\[
   \rk\begin{pmatrix}B\\A\end{pmatrix}=3+1+\rk\begin{pmatrix}0&*&*\\0&0&*\end{pmatrix}.
\]
The rightmost $2\times3$-matrix is the submatrix of~$A$ consisting of those columns having indices $(2,5,7)=\widehat{\sigma}\setminus\lambda$ and row indices $(2,3)$.
The latter is the ordered list $(\alpha_1,\ldots,\alpha_x)$ such that $(\lambda_{\alpha_1},\ldots,\lambda_{\alpha_x})=\lambda\cap\sigma$.
\end{exa}

\begin{proof}[Proof of Lemma \ref{L-AB}]
For any matrix $M$ denote by $M^{(i_1,\ldots,i_x)}_{(j_1,\ldots,j_y)}$ the $x\times y$-submatrix of~$M$ consisting of the rows indexed by $i_1,\ldots,i_x$ and the columns indexed by
$j_1,\ldots,j_y$.
Following the idea of the example, we can clear in the matrix $\SmallMat{B}{A}$  the columns of~$A$ in the pivot positions of~$B$ and observe that
$\rk\SmallMat{B}{A}-b = \rk(A^{(1,\ldots,a)}_{\widehat{\sigma}})$.
Making use of the remaining pivots in~$A$ to clear their respective rows, we see that the rank of $A^{(1,\ldots,a)}_{\widehat{\sigma}}$ equals $|\lambda\setminus\sigma|+\rk(M)$, where
\[
     M=A^{(\alpha_1,\ldots,\alpha_x)}_{\widehat{\sigma}\setminus\lambda}.
\]
Since $|\lambda\setminus\sigma|=a-x$, we conclude that
\[
  \rk\begin{pmatrix}B\\A\end{pmatrix}=b+r\Longleftrightarrow \rk(M)=r-a+x.
\]
Now we obtain the desired result once we show that~$M$ is in $\F[\cF]$ for the stated Ferrers diagram~$\cF$.
From the construction it is clear that the matrix~$M$ is supported by a (top and right aligned) Ferrers diagram.
Thus we just have to count the number of potentially nonzero entries in each column.
The $j^{\text{th}}$ column of~$M$ originates from the column of~$A$ indexed by $\hat{\sigma}_{\beta_j}$, which has the form
\[
   (*\ldots,*,0,\ldots,0)^\top
\]
with a zero at position~$i$ iff $\lambda_i>\hat{\sigma}_{\beta_j}$.
Hence the number of potentially nonzero entries in the $j^{\rm th}$ column of~$M$ is given by
$z_j=|\{i\in[x]\mid \lambda_{\alpha_i}<\hat{\sigma}_{\beta_j}\}|$.
All of this shows that $M\in\F[\cF]$, and this concludes the proof.
\end{proof}

Now we are ready to present explicit expressions for the Krawtchouk coefficients of the pivot partition and its dual.
From Corollary~\ref{C-PpivDual} we know that $\cP^\piv$ and $\cP^\rpiv$ are mutually dual, where $\cP^\rpiv$ is the reverse-pivot partition.
Denote the blocks of the partitions by $P^\piv_\lambda$ and $P^\rpiv_\lambda$, respectively.
Thus
\[
  P^\piv_{\lambda}=\{A\in\Fnm\mid \piv(A)=\lambda\}\quad \text{and} \quad P^\rpiv_{\lambda}=\{A\in\Fnm\mid \rpiv(A)=\lambda\}.
\]

\begin{theo}\label{T-KrawPivot}
Let $\lambda,\,\mu\in\Pi$.
Set $\lambda\setminus\mu=(\lambda_{\alpha_1},\ldots,\lambda_{\alpha_x})$ and
$\mu\setminus\lambda=(\mu_{\beta_1},\ldots,\mu_{\beta_y})$.
Furthermore, for $j\in[y]$ set $z_j=|\{i\in[x]\mid \lambda_{\alpha_i}<\mu_{\beta_j}\}|$ and let $\cF$
be the $x\times y$ Ferrers diagram $\cF=[z_1,\ldots,z_y]$.
Then
\[
   K(\cP^\piv;\lambda,\mu)
   =\sum_{t=0}^m (-1)^{|\lambda|-t} q^{nt+\BinomS{|\lambda|-t}{2}}
     \sum_{r=0}^{|\lambda\cap\widehat{\mu}|}P_{r}(\cF)\GaussianD{|\lambda\cap\widehat{\mu}|-r}{t},
\]
where $(P_r(\cF))_r$ is the rank-weight distribution of~$\F[\cF]$ given in Theorem~\ref{th:new}.
\end{theo}

\begin{proof}
Let $\lambda=(\lambda_1,\ldots,\lambda_a),\,\mu=(\mu_1,\ldots,\mu_c)$, and $\widehat{\mu}=(\sigma_1,\ldots,\sigma_b)$.
By Definition~\ref{D-Kraw}
\[
  K(\cP^\piv;\lambda,\mu)=\sum_{A\in P^\piv_{\lambda}}\chi(\inner{A,C}),\text{ where $C$ is any matrix in }P^\rpiv_{\mu}.
\]
We may use for~$C$ the reverse reduced row echelon form (see Definition~\ref{D-RPivPart}) with reverse pivot indices~$\mu$ and all unspecified entries equal to zero.
Thus, using the standard basis vectors $e_i\in\F^n$ we may choose
\[
  C=(C_1,\ldots,C_m)\in\F^{n\times m}\text{ where }
  C_j=\left\{\begin{array}{cl} e_{\alpha} &\text{if }j=\mu_{\alpha} \text{ for some }\alpha\in\{1,\ldots,c\}, \\ 0 &\text{else}.\end{array}\right.
\]
Set $V=\rs(C)$. We also need $V^\perp$, which is given by $V^\perp=\rs(B)$, where
\[
  B=(B_1,\ldots,B_m)\in\F^{n\times m}\text{ where }
  B_j=\left\{\begin{array}{cl} e_{\beta} &\text{if }j=\sigma_{\beta},\\ 0 &\text{else.}\end{array}\right.
\]
Note that $\dim V^\perp=b=m-c$.
In the following computation we make use of the Krawtchouk coefficients for the row-space partition, which have been determined in Theorem~\ref{T-KrawRS}.
Using that any subspace~$U$ with $\piv(U)=\lambda$ satisfies $\dim(U)=|\lambda|=a$ we compute
\begin{eqnarray*}
\allowdisplaybreaks
  K(\cP^\piv;\lambda,\mu)&=&\sum_{\substack{U\in\cL\\\piv(U)=\lambda}}\sum_{\substack{A\in\Fnm\\\rs(A)=U}}\chi(\inner{A,C}) \\
      &=&\sum_{\substack{U\in\cL\\\piv(U)=\lambda}}K(\cP^\rs;U,V)\\
      &=&\sum_{\substack{U\in\cL\\\piv(U)=\lambda}}\sum_{t=0}^m(-1)^{a-t}q^{nt+\BinomS{a-t}{2}}\GaussianD{\dim(U\cap V^\perp)}{t}\\
      &=&\sum_{t=0}^m(-1)^{a-t}q^{nt+\BinomS{a-t}{2}}\sum_{\substack{U\in\cL\\\piv(U)=\lambda}}\GaussianD{\dim(U\cap V^\perp)}{t}\\
      &=&\sum_{t=0}^m(-1)^{a-t}q^{nt+\BinomS{a-t}{2}}\sum_{\substack{A\in\Fnm\text{ in RREF}\\\piv(A)=\lambda}}\GaussianD{\dim(\rs(A)\cap V^\perp)}{t}.
\end{eqnarray*}
It remains to determine the inner sum. Since $V^\perp=\rs(B)$, we conclude that
$\dim(\rs(A)\cap V^\perp)=a+b-\dim(\rs(A)+V^\perp)=a+b-\rk\SmallMat{A}{B}$.
As mentioned after Lemma~\ref{L-AB}, for any matrix~$A$ as specified we have $\rk\SmallMat{A}{B}\in\{b+r\mid r=a-x,\ldots,a\}$,
where $x=|\lambda\cap\widehat{\mu}|$.
Thus, thanks to the lemma the inner sum turns into
\[
  \sum_{r=a-x}^{a}\sum_{\substack{A\in\Fnm\text{ in RREF}\\\piv(A)=\lambda, \; \rk\SmallMat{A}{B}=b+r}}\GaussianD{a-r}{t}
  =\sum_{r=a-x}^{a}P_{r-a+x}(\cF)\GaussianD{a-r}{t}=\sum_{r=0}^{x}P_{r}(\cF)\GaussianD{x-r}{t}
\]
with~$\cF$ as in the theorem.
This concludes the proof.
\end{proof}

With the aid of Theorem~\ref{th:MWI} we now obtain an explicit MacWilliams identity for the pivot distributions of matrix codes and the reverse pivot distribution of the dual codes 
by substituting the Krawtchouk coefficients from Theorem~\ref{T-KrawPivot}.
We omit the resulting explicit form of the identities.

\section{Pivot-Extremal Codes}\label{S-PivotExtr}
In this section we generalize the notion of extremality to the pivot partition.
To do so, we need to introduce a partial order on the set $\Pi$ of all possible pivot lists for matrices in~$\Fnm$.
This is done in the obvious way:
for $\lambda,\,\mu\in\Pi$ define $\lambda\leq\mu$ if $\lambda \subseteq \mu$,
where for the latter we simply interpret pivot lists as sets. Then $(\Pi,\leq)$ is a lattice, which of course is isomorphic to the subset lattice of~$[m]$.
Recall that $n\geq m$.

The following results from basic Linear Algebra will be crucial.

\begin{lemma}\label{P-PiLatt}
\begin{enumerate}
\item Let $U,V\in\cL$ such that $U\leq V$. Then $\piv(U)\leq\piv(V)$.
\item Let $\lambda,\,\mu\in\Pi$ such that $\mu\leq\lambda$ and let $V\in\cL$ be such that $\piv(V)=\lambda$.
        Then there exists $U\in\cL$ such that $\piv(U)=\mu$ and $U\leq V$.
\item Let $\lambda,\,\mu\in\Pi$ such that $\mu\leq\lambda$ and let $U\in\cL$ be such that $\piv(U)=\mu$.
       Then there exists $V\in\cL$ such that $\piv(V)=\lambda$ and $U\leq V$.
\end{enumerate}
\end{lemma}

\begin{proof}
1) We may write $U=\rs(A)$ and $V=\rs(M)$, where $M=\SmallMat{A}{B}$. Then $\piv(U)\leq\piv(V)$ follows from~\eqref{e-pivcol} applied to $A$ and~$M$.

2) Let $V=\rs(A)$, where $A\in\Fnm$ is in RREF. Hence $\piv(A)=\lambda\supseteq\mu$.
Let~$B$ be the submatrix of~$A$ consisting of the rows of~$A$ that contain the pivots in~$\mu$.
Then $U:=\rs(B)\leq V$ and $\piv(U)=\mu$.

3) Let $A\in\Fnm$ be in RREF such that $U=\rs(A)$. Then $\piv(A)=\mu$. Let $\lambda\setminus\mu=(\sigma_1,\ldots,\sigma_\ell)$.
Consider the matrix $M=\SmallMat{A}{B}$, where
\[
    B=\begin{pmatrix}e_{\sigma_1}\\ \vdots\\ e_{\sigma_\ell}\end{pmatrix},
\]
where $e_i$ denotes the standard basis row vectors in~$\F^m$.
Then $\piv(M)=\mu\cup(\lambda\setminus\mu)=\lambda$ and $V=\rs(M)$ is the desired subspace. \qedhere
\end{proof}

We now define the analogue of~$\cC(U)$ from Definition~\ref{D-CodeEnum} for the pivot partition and reverse-pivot partition.
The following properties are immediate with Lemma~\ref{P-PiLatt} and~\ref{bound}.

\begin{prop}\label{D-Clambda}
Let $\cC\leq\Fnm$ be a code and $\lambda\in\Pi$. Then
\[
   \cC(\lambda,\piv):=\{A\in\cC\mid \piv(A)\leq\lambda\}=\bigcup_{\substack{U\in\cL\\ \piv(U)=\lambda}}\cC(U).
\]
In particular, $\cC(\lambda,\piv)=\{0\}$ if and only if $\cC(U)=\{0\}$ for all $U\in\cL$ with $\piv(U)=\lambda$.
Thus, if~$\cC(\lambda,\piv)=\{0\}$, then $|\cC|\leq q^{n(m-|\lambda|)}$.
Note that $\cC(\lambda,\piv)$ is \emph{not} a subspace in general.
Likewise we define
\[
     \cC(\lambda,\rpiv)=\{A\in\cC\mid \rpiv(A)\leq\lambda\}=\bigcup_{\substack{U\in\cL\\ \rpiv(U)=\lambda}}\cC(U),
\]
which has the analogous properties.
\end{prop}

This gives naturally rise to the following notion of extremal codes.

\begin{defi}\label{D-lambdaExt}
Let $\lambda\in\Pi$. A code $\cC\leq\Fnm$ is called \textbf{$(\lambda,\piv)$-extremal} if $\cC(\lambda,\piv)=\{0\}$ and $|\cC|=q^{n(m-|\lambda|)}$.
A code that is $(\lambda,\piv)$-extremal for some $\lambda\in\Pi$ is called \textbf{$\piv$-extremal}.
According definitions are in place for $(\lambda,\rpiv)$.
\end{defi}

Therefore
\begin{equation}\label{e-extremal}
  \cC\text{ is $(\lambda,\piv)$-extremal }\Longleftrightarrow \cC\text{ is $U$-extremal for all $U\in\cL$ with }\piv(U)=\lambda.
\end{equation}

\begin{rem}\label{R-MRDlambda}
Let~$\cC$ be a nonzero code. Then
\[
  \cC\text{ is MRD with minimum distance $d$}\Longleftrightarrow \cC\text{ is $(\lambda,\piv)$-extremal for all $\lambda$ such that }|\lambda|=d-1.
\]
The forward direction is immediate with Remark~\ref{R-MRDExtr2}.
For the backward direction note that $|\cC|=q^{n(m-d+1)}$ by assumption, and the distance is clearly not smaller than~$d$.
\end{rem}

Just like for the rank-weight and the subspace distribution, extremality is preserved under dualization. This is an immediate consequence of Propositions~\ref{dualis} and~\ref{P-VperpPiv}.

\begin{prop}\label{P-PivExtrDual}
Let $\cC\leq\Fnm$ and $\lambda\in\Pi$.
Then $\cC$ is $(\lambda,\piv)$-extremal iff $\cC^\perp$ is $(\widehat{\lambda},\rpiv)$-extremal.
\end{prop}

As for $U$-extremal codes, the partial pivot partition distribution of pivot-extremal codes satisfies some rigidity properties.
Its values depend on the cardinality of the blocks $P^\piv_\mu$ of the pivot partition, which we therefore compute first.

\begin{prop}\label{P-PpivCard}
Let $\mu=(\mu_1,\ldots,\mu_r)\in\Pi$.
Define
\[
  c(\mu)=\sum_{i=1}^r(m-\mu_i-r+i).
\]
Then
$|P^\piv_{\mu}|=q^{c(\mu)}\prod_{i=0}^{r-1}(q^n-q^i)$. Note also that $|P^\piv_{\mu}|=1$ if $\mu=()$, the empty list.
\end{prop}

\begin{proof}
Consider a matrix in reduced row echelon form with pivot list~$\mu$.
The number of unspecified entries in the $i^{\text{th}}$ row is given by $m-\mu_i-(r-i)$.
This shows that there exist $q^{c(\mu)}$ matrices $A\in\F^{n\times m}$ in RREF with $\piv(A)=\mu$.
Denote the set of these matrices by $R(\mu)$. Then $|R(\mu)|=q^{c(\mu)}$.

The partition set $P^\piv_{\mu}$ is the set of all matrices in $\Fnm$ with pivot list~$\mu$.
It thus forms the disjoint union of the orbits of the matrices in~$R(\mu)$ under the group action
\[
   \GL_n(\F)\times\Fnm\longrightarrow\Fnm,\quad (X,A)\longmapsto XA.
\]
In order to determine the orbit size of any $A\in R(\mu)$, we use the orbit-stabilizer theorem.
For $A\in R(\mu)$ we have $A=\SmallMat{\hat{A}}{0}$, where $\hat{A}\in\F^{r\times m}$ has full row rank.
This tells us that for any matrix
\[
    X=\begin{pmatrix}X_1&X_2\\X_3&X_4\end{pmatrix}\in\GL_n(\F), \text{ where } X_1\in\F^{r\times r},
\]
we have $XA=A$ iff $X_3=0$ and $X_1=I_r$.
Hence~$X_2,X_4$ are free and thus the stabilizer has cardinality $q^{r(n-r)}|\GL_{n-r}(\F)|$.
Now we arrive at
\[
    |P^\piv_\mu|=q^{c(\mu)}\frac{|\GL_n(\F)|}{q^{r(n-r)}|\GL_{n-r}(\F)|}
   =q^{c(\mu)}\frac{\prod_{i=0}^{n-1}(q^n-q^i)}{q^{r(n-r)}\prod_{i=0}^{n-r-1}(q^{n-r}-q^i)}
  =q^{c(\mu)}\prod_{i=0}^{r-1}(q^n-q^i),
\]
as desired.
\end{proof}

Now we can formulate the main result of this section.

\begin{theo}\label{T-lambdaExtrPD}
Let $\lambda \in \Pi$ and let $0 \le u \le |\lambda|$ be an integer. Suppose that a code $\cC$ is $(\lambda',\piv)$-extremal for all $\lambda' \le \lambda$ with $|\lambda'|=u$.
Then for all $\mu \le \lambda$ we have
\[
  \cP^\piv(\cC,\mu)
  =q^{c(\mu)}\prod_{i=0}^{|\mu|-1}(q^n-q^i)\bigg(\sum_{i=0}^u \GaussianD{|\mu|}{i}
   {(-1)}^{|\mu|-i} q^{\binom{|\mu|-i}{2}} + \sum_{i=u+1}^{|\mu|} \GaussianD{|\mu|}{i} q^{n(i-u)} {(-1)}^{|\mu|-i} q^{\binom{|\mu|-i}{2}}\bigg),
\]
where $c(\mu)$ is defined as in Proposition~\ref{P-PpivCard}.
Thus, the partial pivot distribution of~$\cC$ depends only on $q,n,u$.
\end{theo}

\begin{proof}
Note first that by~\eqref{e-extremal} the assumptions imply that~$\cC$ is $U$-extremal for all subspaces $U\in\cL$ such that $\dim(U)=u$ and $\piv(U)\leq\lambda$.
Next, by definition, we have
\[
   \cP^\piv(\cC,\mu)=\big|\{A\in\cC\mid \piv(A)=\mu\}\big|=\Big|\bigcup_{\piv(T)=\mu}\{A\in\cC\mid \rs(A)=T\}\Big|
    =\sum_{\piv(T)=\mu}\cP^\rs(\cC,T).
\]
Fix any~$\mu$ such that $\mu\leq\lambda$. The case $\mu=()$ is trivial.
If $0<|\mu|\leq u$, then the right hand side of the formula in the theorem is~$0$.
This is indeed $\cP^\piv(\cC,\mu)$ because, thanks to Lemma~\ref{P-PiLatt}(3), any subspace~$T$ with $\piv(T)=\mu$ is contained in a subspace~$U$
such that $\mu\leq\piv(U)\leq\lambda$ and $\dim(U)=u$.
Thus $\cC(T)\le \cC(U)$ and $U$-extremality implies $\cP^\rs(\cC,T)=0$.

Let now $|\mu|>u$.
Fix a subspace~$T$ such that $\piv(T)=\mu$.
Let $U \le T$ be an arbitrary subspace of dimension~$u$ and let $\lambda'=\piv(U)$.
Clearly, $|\lambda'|=u$. Since $U \le T$, Lemma~\ref{P-PiLatt}(1) implies
$\lambda' \le \piv(T)=\mu \le \lambda$.
Thus~$\cC$ is $U$-extremal.

All of this shows that $\cC$ is $U$-extremal for any subspace $U\leq T$ of dimension~$u$.
In other words,~$\cC$ satisfies the assumptions of Theorem~\ref{T-Rigid}.
Since this is the case for any subspace~$T$ such that $\piv(T)=\mu$ we conclude
\[
  \cP^\piv(\cC,\mu)=\sum_{\piv(T)=\mu}\bigg(\sum_{i=0}^u \GaussianD{|\mu|}{i}
   {(-1)}^{|\mu|-i} q^{\binom{|\mu|-i}{2}} + \sum_{i=u+1}^{|\mu|} \GaussianD{|\mu|}{i} q^{n(i-u)} {(-1)}^{|\mu|-i} q^{\binom{|\mu|-i}{2}}\bigg).
\]
Since the summands do not depend on the specific choice of~$T$, we arrive at
\[
     \cP^\piv(\cC,\mu)=|P^\piv_{\mu}|\bigg(\sum_{i=0}^u \GaussianD{|\mu|}{i}
   {(-1)}^{|\mu|-i} q^{\binom{|\mu|-i}{2}} + \sum_{i=u+1}^{|\mu|} \GaussianD{|\mu|}{i} q^{n(i-u)} {(-1)}^{|\mu|-i} q^{\binom{|\mu|-i}{2}}\bigg),
\]
and Proposition~\ref{P-PpivCard} concludes the proof.
\end{proof}

We conclude this section with an example of a code~$\cC$ that satisfies the assumptions of above theorem, but is not MRD.

\begin{exa}\label{E-PivRigidNotMRD}
Let $m=m_1+m_2$ with $m_1 \ge 1$ and $m_2 \ge 2$. Let $n \ge m$ and $\lambda=(m_1+1,...,m_1+m_2)$. Fix $1 \le u \le m_2-1$. Let $\cC_2 \le \F^{n \times m_2}$ be an MRD code of minimum distance $u+1$. Construct the code
$$\mC:=\{(A \mid B) \in \F^{n \times m} \mid A \in \F^{n \times m_1}, \; B \in \mC_2\}.$$
Then $\cC$ has minimum distance 1 and cardinality $|\cC|=q^{n(m-u)}$. In particular, $\cC$ is not MRD. We claim that $\cC$ is $(\lambda',\piv)$-extremal for all $\lambda' \le \lambda$ with $|\lambda'|=u$. Fix any $\lambda' \le \lambda$ with $|\lambda'|=u$, and let $U \le \F^m$ be any space with $\piv(U)=\lambda'$.
There is only one space $V \le \F^m$ with $\piv(V)=\lambda$, namely,
$V= \langle e_{m_1+1},...,e_{m} \rangle$, where $\{e_1,...,e_m\}$ is the canonical basis of $\F^m$. It is easy to see that $U \le V$. Since $\mC_2$ is MRD with minimum distance $u+1$, we have $\mC(U)=\{0\}$. As $|\cC|=q^{n(m-u)}$,
$\mC$ is $(\lambda',\piv)$-extremal for all $\lambda' \le \lambda$ with $|\lambda'|=u$, as claimed.
\end{exa}

\section{Matrices Supported on Ferrers Diagrams and $q$-Rook Polynomials}
\label{sec:FD}

In this section we explicitly compute the rank distribution of matrices supported on an arbitrary Ferrers diagram $\cF$, establishing Theorem \ref{th:new}. In particular, we prove that
$P_r(\cF)$ is a polynomial in $q$ for every value of $r$ and every diagram $\cF$.
We then exploit connections between the rank distribution of matrices supported on a Ferrers diagram and $q$-rook polynomials, giving explicit formulas for these and establishing the monotonicity in $r$ of $\deg(P_r(\cF))$.
We follow the notation of Definitions~\ref{D-Ferrers} and~\ref{D-FMat} and~\eqref{e-Irm},~\eqref{e-len}. In this section we do \textit{not} assume $m \le n$.

For $r \in \N_0$ let
$P_r(c_1,...,c_m):=P_r(\cF)$, where $\cF=[c_1,...,c_m]$ is the Ferrers diagram whose columns lengths are
$c_1\leq\ldots\leq c_m$. Then $P_0(c_1,\ldots,c_m)=1$ and $P_r(c_1,...,c_m)=0$ for $r>c_m$.

\begin{theo}\label{T-PrExplForm}
Let  $c_1,...,c_m$ be integers with $c_{i+1} \ge c_i$ for all $i$.
\begin{enumerate}
\item For $r\in\N$ we have the recursion
        \[
            P_r(c_1,\ldots,c_m)=P_{r-1}(c_1,\ldots,c_{m-1})(q^{c_m}-q^{r-1})+P_r(c_1,\ldots,c_{m-1})q^r
        \]
        with initial conditions
        \[
             P_0(c_1,\ldots,c_s)=1 \mbox{ for all $s$}, \quad P_1(c_1)=q^{c_1}-1,\quad P_r(c_1)=0\text{ for }r>1.
        \]
\item Let $r\in\N_0$. Then~$P_r(c_1,...,c_m)$ is given by the explicit formula
        \begin{equation}\label{e-ExplForm}
           P_r(c_1,\ldots,c_m)=\sum_{(i_1,\ldots,i_r)\in\cI_{r,m}}q^{rm-\len{i}}\prod_{j=1}^r(q^{c_{i_j}-j+1}-1).
        \end{equation}
\end{enumerate}
\end{theo}

\begin{proof}
1) The initial conditions are clear.
Furthermore, both sides of the recursion are zero  if $r>c_m$.
Thus let~$r\leq c_m$.
Consider a matrix $M\in\F[\cF]$ of rank~$r$. Denote the submatrix of~$M$ consisting of the first $m-1$ columns by~$\hat{M}$.
If~$\hat{M}$ has rank~$r-1$, then the last column of~$M$ can be any choice outside the column span of~$\hat{M}$.
Since $c_m\geq c_i$ for all~$i$, this results in $q^{c_m}-q^{r-1}$ options.
If~$\hat{M}$ has rank~$r$, then the last column of~$M$ has to be in the column span of~$\hat{M}$.
This results in~$q^r$ options.
This proves the desired recursion.

2)  First of all,~\eqref{e-ExplForm} is satisfied if $r>m$ because then
$\cI_{r,m}=\emptyset$.  It is also trivially true for $r=0$ and all $m$.

We now proceed by induction on~$r$.
\allowdisplaybreaks
Assume~\eqref{e-ExplForm} for all ranks at most~$r-1$ and all~$m$.
We want to show the identity for rank~$r$ and all $m$.
To do so, we induct on~$m$.
The induction hypothesis is provided by all $m\in\{1,\ldots,r-1\}$, in which case both sides of~\eqref{e-ExplForm} are zero.
Thus let $m\geq r$.
Denote the right hand side of~\eqref{e-ExplForm} by~$Q$.
We now show that $Q-P_{r-1}(c_1,\ldots,c_{m-1})(q^{c_m}-q^{r-1})=P_r(c_1,\ldots,c_{m-1})q^r$.
Thanks to the recursion in~(1), which is true for all~$r$ regardless of~$m$, this establishes $Q=P_r(c_1,\ldots,c_m)$.

We compute $Q-P_{r-1}(c_1,\ldots,c_{m-1})(q^{c_m}-q^{r-1})=$
\begin{align*}
    &=\sum_{i\in\cI_{r,m}}q^{rm-\len{i}}\prod_{j=1}^r(q^{c_{i_j}-j+1}-1)
    -\sum_{i\in\cI_{r-1,m-1}}q^{(r-1)(m-1)-\len{i}}\prod_{j=1}^{r-1}(q^{c_{i_j}-j+1}-1)(q^{c_m}-q^{r-1})\\
    &=\sum_{i\in\cI_{r-1,m-1}}\prod_{j=1}^{r-1}(q^{c_{i_j}-j+1}-1)\Big[\sum_{i_r=i_{r-1}+1}^m q^{rm-\len{i}-i_r}(q^{c_{i_r}-r+1}-1)
    - q^{(r-1)(m-1)-\len{i}}(q^{c_m}-q^{r-1})\Big]\\
    &=\sum_{i\in\cI_{r-1,m-1}}\prod_{j=1}^{r-1}(q^{c_{i_j}-j+1}-1)\Big[\sum_{i_r=i_{r-1}+1}^m q^{rm-\len{i}-i_r}(q^{c_{i_r}-r+1}-1)-q^{(r-1)m-\len{i}}(q^{c_m-r+1}-1)\Big] \\
    &=\sum_{i\in\cI_{r-1,m-1}}\prod_{j=1}^{r-1}(q^{c_{i_j}-j+1}-1)\Big[\sum_{i_r=i_{r-1}+1}^{m-1}q^{rm-\len{i}-i_r}(q^{c_{i_r}-r+1}-1)\Big]\\
    &=\sum_{i\in\cI_{r,m-1}}q^{rm-\len{i}}\prod_{j=1}^r(q^{c_{i_j}-j+1}-1)=P_r(c_1,\ldots,c_{m-1})q^r.
\end{align*}
This establishes~\eqref{e-ExplForm} and concludes the proof. \qedhere
\end{proof}

For the rest of this section we regard~$q$ as an indeterminate over~$\Z$.
Thus $\Z[q]$ (resp.\ $\Z[q,q^{-1}]$), denotes the ring of polynomials (resp.\ Laurent polynomials) in~$q$ with integer coefficients.
From~\eqref{e-ExplForm} it is clear that we may regard $P_r(\cF)$ as an element in $\Z[q,q^{-1}]$.
We have the following result.

\begin{cor} \label{cor:deg}
Let $\cF=[c_1,...,c_m]$ be a Ferrers diagram and $r\in\N_0$. Then $P_r(\cF)\in\Z[q]$.
Moreover, set
$$\cI_{r,m}(\cF):=\{ i \in \cI_{r,m} \mid c_{i_j} \neq j-1 \mbox{ for all }1 \le j \le r\}.$$
Then
$$\deg(P_r(\cF)) = \left\{ \begin{array}{cl} - \infty & \mbox{if } \cI_{r,m}(\cF)=\emptyset, \\ rm-\binom{r}{2} +\max\left\{ \sum_{j=1}^r (c_{i_j}-i_j) \mid i \in \cI_{r,m}(\cF) \right\} & \mbox{if } \cI_{r,m}(\cF) \neq\emptyset. \end{array}\right.$$
\end{cor}

\begin{proof}
The polynomiality of~$P_r(\cF)$ follows from the recursion and initial conditions in Theorem~\ref{T-PrExplForm}(1).
Alternatively, one can derive this fact from~\eqref{e-ExplForm} by
verifying that if in one of the rightmost products there occurs a negative exponent of~$q$, then the product is actually zero.
More precisely, if $c_{i_j}<j-1$, then there exists some $\ell<j$ such that $c_{i_\ell}=\ell-1$. This follows indeed easily from
$0\leq c_1\leq\ldots\leq c_m$.

As for the degree, consider again~\eqref{e-ExplForm}.
Clearly, the summands corresponding to $i\not\in\cI_{r,m}(\cF)$ are zero.
Furthermore, for any $i\in\cI_{r,m}(\cF)$ the degree of the corresponding summand is
$rm-\len{i}+\sum_{j=1}^r(c_{i_j}-j+1)=rm-\binom{r}{2}+\sum_{j=1}^r(c_{i_j}-i_j)$.
\end{proof}

Note that the argument in the first paragraph also shows that the set $\cI_{r,m}$ is in fact given by
$\cI_{r,m}(\cF)=\{ i \in \cI_{r,m} \mid c_{i_j} > j-1 \mbox{ for all }1 \le j \le r\}$.

\begin{rem}
We wish to point out that for certain matrix spaces $\F[\cF]$ where~$\cF\subseteq[n]\times[m]$ is \emph{not} a Ferrers diagram, the rank-weight functions $P_r(\cF)$ are not necessarily polynomials in~$q$.
The smallest known case is for $n=m=7$ and where~$\cF$ is the support of the point-line incidence matrix of the Fano plane, see~\cite[Sec.~1]{KLM14} and \cite[p.~381]{Ste98}.
\end{rem}

The formula in Theorem \ref{T-PrExplForm} takes a simpler form for some highly regular diagrams.
This is the case, for example, for the upper-triangular board.
The following is easily verified.

\begin{cor}\label{C-PrTriang}
Let $\cF=[1,2,\ldots,m]$ be the $m\times m$-upper triangle. Then
        \[
          P_r(1,\ldots,m)=\sum_{i\in\cI_{r,m}}\prod_{j=1}^r(q^{m-j+1}-q^{m-i_j}) \quad \text{for all } r\in\N_0.
        \]
\end{cor}

\begin{rem}\label{R-Rectangle}
Let $\cF=[n,\ldots,n]$ be the $n\times m$ rectangle. For all $r\in\N_0$ we have
\begin{align*}
  P_r(n,\ldots,n)&=\sum_{i\in\cI_{r,m}} q^{rm-\len{i}}\prod_{j=1}^r(q^{n-j+1}-1)
    =q^{rm-\BinomS{r}{2}}\sum_{i\in\cI_{r,m}}q^{-\len{i}}\prod_{j=0}^{r-1}(q^n-q^j)\\
   & =q^{-\BinomS{r}{2}}\sum_{i\in\cI_{r,m}}q^{\sum_{j=1}^r (m-i_j)}\prod_{j=0}^{r-1}(q^n-q^j)
    =q^{-\BinomS{r}{2}}\sum_{0\leq t_1<\ldots<t_r\leq m-1}q^{\sum_{j=1}^r t_j}\prod_{j=0}^{r-1}(q^n-q^j).
\end{align*}
Comparing coefficients in the $q$-binomial identity
$\sum_{r=0}^mq^{\BinomS{r}{2}}\Gaussian{m}{r}t^r=\prod_{j=0}^{m-1}(1+q^jt)$,
one easily verifies that the last expression above simplifies to
 $\Gaussian{m}{r}\prod_{j=0}^{r-1}(q^n-q^j)$,
which is indeed known as the number of matrices in~$\F^{n\times m}$ of rank~$r$.
\end{rem}

Following work by Solomon \cite{solomon1990bruhat}, Haglund in \cite[Section 2]{Hag98} establishes an interesting connection between $P_r(\cF)$ and the $q$-rook polynomial $R_r(\cF)$
for an arbitrary Ferrers board $\cF=[c_1,\ldots,c_m]$.
The latter has been introduced by Garsia/Remmel~\cite{GaRe86} and is defined as
follows.

\begin{defi}\label{D-Rook}
The $q$-rook polynomial associated with $\cF$ and $r\in\N_0$ is defined as
$$R_r(\cF)= \sum_{C \in \textnormal{NAR}_r(\cF)} q^{\textnormal{inv}(C,\cF)} \in \Z[q],$$
where $\textnormal{NAR}_r(\cF)$ is the set of all placements of $r$ non-attacking rooks on $\cF$ (non-attacking means that no two rooks are in
the same column, and no two are in the same row), and $\textnormal{inv}(C,\cF) \in \N_0$ is computed as follows:
For a placement $C$, cross out all dots which either contain a rook, or are
above or to the right of any rook; then $\textnormal{inv}(C, \cF)$ is the number of dots of $\cF$ not crossed out.
\end{defi}

For instance, placing on $\cF=[1,2,4,4,5]$ the following three rooks (R) results in $\text{inv}(C,\cF)=7$.
\begin{figure}[ht]
    \centering
     {\small
     \begin{tikzpicture}[scale=0.4]
         \draw (4.5,1.5) node (b1) [label=center:$\bullet$] {};
         \draw (4.5,2.5) node (b1) [label=center:$\times$] {};
         \draw (4.5,3.5) node (b1) [label=center:$\bullet$] {};
         \draw (4.5,4.5) node (b1) [label=center:R] {};
         \draw (4.5,5.5) node (b1) [label=center:$\times$] {};
       \

         \draw (3.5,2.5) node (b1) [label=center:$\times$] {};
         \draw (3.5,3.5) node (b1) [label=center:$\bullet$] {};
         \draw (3.5,4.5) node (b1) [label=center:$\bullet$] {};
         \draw (3.5,5.5) node (b1) [label=center:R] {};

         \draw (2.5,2.5) node (b1) [label=center:R] {};
         \draw (2.5,3.5) node (b1) [label=center:$\times$] {};
         \draw (2.5,4.5) node (b1) [label=center:$\times$] {};
         \draw (2.5,5.5) node (b1) [label=center:$\times$] {};

         \draw (1.5,4.5) node (b1) [label=center:$\bullet$] {};
        \draw (1.5,5.5) node (b1) [label=center:$\bullet$] {};

       \draw (0.5,5.5) node (b1) [label=center:$\bullet$] {};
     \end{tikzpicture}
     }
     \end{figure}

Note that $|\cF|$ is the number of dots in~$\cF$. Thus $|\cF|=\sum_{j=1}^m c_j$ for $\cF=[c_1,\ldots,c_m]$.
Hence $R_0(\cF)=q^{|\cF|}$ for any Ferrers diagram $\cF$, including the empty diagram.
Furthermore, note that $R_r(\cF)$ is the zero polynomial if and only if $\textnormal{NAR}_r(\cF)=\emptyset$.

The connection between $q$-rook polynomials and the distribution of matrices supported on $\cF$ lies in the following elegant formula by Haglund.

\begin{theo}[\mbox{\hspace*{-.4em}\cite[Thm.~1]{Hag98}}] \label{th:Hag}
For any Ferrers diagram $\cF$and any $r\in\N_0$ we have
$$P_r(\cF)=(q-1)^r \; q^{|\cF|-r} \; R_r(\cF)_{|q^{-1}}$$
in the ring $\Z[q,q^{-1}]$.
\end{theo}

Combining Theorems~\ref{T-PrExplForm}  and \ref{th:Hag} we obtain an explicit formula for the $q$-rook polynomials.
Examples of $R_r(\cF)$ for some Ferrers diagrams are listed in~\cite[pp.~273]{GaRe86}.

\begin{cor} \label{cor:expl}
For any Ferrers diagram $\cF=[c_1,...,c_m]$ and for any $r\in\N_0$ we have
\[
   R_r(\cF)=\frac{q^{\sum_{j=1}^m c_j-rm}\sum_{i\in\cI_{r,m}}\prod_{j=1}^r(q^{i_j+j-c_{i_j}-1}-q^{i_j})}{(1-q)^r}.
\]
\end{cor}

\begin{rem}
Corollary \ref{cor:expl} can be used to derive an explicit formula for the $q$-Stirling number of the second kind. The latter are defined via the recursion
\[
  S_{m+1,r}=q^{r-1}S_{m,r-1}+\frac{q^r-1}{q-1}S_{m,r}
\]
with initial conditions $S_{0,0}(q)=1$ and $S_{m,r}(q)=0$ for $r<0$ or $r>m$.\footnote{In the combinatorics literature $q$-Stirling number of the second kind are often defined via the
recursion $\tilde{S}_{m+1,r}(q)=\tilde{S}_{m,r-1}(q)+(q^r-1)/(q-1)\tilde{S}_{m,r}(q)$. It is easily seen that $S_{m,r}(q)=q^{\BinomS{r}{2}}\tilde{S}_{m,r}(q)$.} It
is known \cite[p. 248]{GaRe86} that for all $m$ and $r$ we have
\[
   S_{m+1,m+1-r}=R_r(\cF),
 \]
 where $\cF=[1,...,m]$ is the upper-triangular $m \times m$ Ferrers board.
 Therefore applying Corollary~\ref{cor:expl} we obtain
 \[
   S_{m+1,m+1-r}=\frac{q^{\BinomS{m+1}{2}-rm}\sum_{i\in\cI_{r,m}}\prod_{j=1}^r(q^{j-1}-q^{i_j})}{(1-q)^r} \quad \text{for }
   0\leq r\leq m+1.
 \]
\end{rem}

As a second application of Theorem \ref{T-PrExplForm}, we recover the recursion shown in \cite{GaRe86} for the $q$-rook polynomials $R_r(\cF)$.

\begin{cor}[see also \text{\cite[Thm. 1.1]{GaRe86}}]
Let $\cF=[c_1,...,c_m]$ be a Ferrers diagram, and let $\cF'=[c_1,...,c_{m-1}]$.
For all $r \ge 1$ we have
\[
   R_r(\cF) = R_r(\cF') \; q^{c_m-r} + R_{r-1}(\cF') \; \frac{q^{c_m-r+1}-1}{q-1}.
\]
\end{cor}

\begin{proof}
By Theorem \ref{th:Hag} we have
\[
   R_r(\cF)_{|q^{-1}} = P_r(\cF) \; q^{r-|\cF|} \; (q-1)^{-r}.
\]
Using the recursion for $P_r(\cF)$ established in Theorem \ref{T-PrExplForm}
we obtain
\[
   R_r(\cF)_{|q^{-1}} = \Big( P_{r-1}(\cF') \; (q^{c_m}-q^{r-1}) + P_r(\cF') \; q^r \Big) q^{r-|\cF|} \; (q-1)^{-r}.
\]
Using that $|\cF'|=\sum_{j=1}^{m-1}c_j$ and $|\cF|=|\cF'|+c_m$ and applying Theorem \ref{th:Hag} twice we arrive at
\[
    R_r(\cF)_{|q^{-1}} = (q-1)^{-1} \; q^{-c_m+1} \; R_{r-1}(\cF')_{|q^{-1}} \; (q^{c_m}-q^{r-1}) + q^{-c_m} \; R_r(\cF')_{|q^{-1}} \; q^r.
\]
Applying the transformation $q \longmapsto q^{-1}$ yields the desired result.
\end{proof}

We conclude this section by studying the degree (in $q$) of the polynomials $P_r(\cF)$.
We will show that, for any given diagram $\cF$, the function $r \longmapsto \deg(P_r(\cF))$ is strictly increasing as long as $P_r(\cF)\not\equiv0$.
This fact does not seem obvious from the explicit expression for $\deg(P_r(\cF))$ given in Corollary \ref{cor:deg}.
Therefore we take a different approach based on rook placements.
This will also give us the chance to establish new connections between $P_r(\cF)$ and $R_r(\cF)$.

Recall that the \textbf{trailing degree} of a Laurent polynomial
$$P=\sum_{i}a_i q^i \in \Z[q,q^{-1}]$$ is defined as
$\textnormal{tdeg}(P)= \min\{i \mid a_i \neq 0\}$.
The trailing degree of the zero polynomial is $+\infty$ by definition. Moreover,
for any (possibly zero) Laurent polynomial $P \in \Z[q,q^{-1}]$ one has
\begin{equation} \label{degtdeg}
\deg\left(P_{|q^{-1}}\right)=-\textnormal{tdeg} (P).
\end{equation}

We can relate the degree of $P_r(\cF)$ and the trailing degree of $R_r(\cF)$ as follows.

\begin{prop} \label{prop:tdeg}
Let $\cF$ be a Ferrers diagram, and let $r \ge 0$. We have
$$\deg(P_r(\cF))=|\cF| - \textnormal{tdeg}(R_r(\cF)).$$
In particular, $P_r(\cF)$ is the zero polynomial if and only if $R_r(\cF)$ is the zero polynomial.
\end{prop}

\begin{proof}
By Theorem \ref{th:Hag} we have the identity
\begin{equation*} \label{cmp}
q^r P_r(\cF) = (q-1)^r \; q^{|\cF|} \; R_r(\cF)_{|q^{-1}}
\end{equation*}
in $\Z[q,q^{-1}]$. Taking degrees and using (\ref{degtdeg}) we obtain $r+\deg(P_r(\cF))=r+|\cF|-\textnormal{tdeg}(R_r(\cF))$.
\end{proof}

We can finally show that the function $r \longmapsto \deg(P_r(\cF))$ is strictly increasing on $[0,\overline{r}]$, where $\overline{r}$ is the maximum $r$ with
$P_r(\cF) \neq 0$. The proof relies on Proposition \ref{prop:tdeg} and on the following preliminary result.

\begin{lemma} \label{lem:tec}
Let $\cF$ be a Ferrers diagram, and let $r \ge 1$. If $\textnormal{tdeg}(R_{r}(\cF))=0$, then $R_{r+1}(\cF)=0$.
\end{lemma}

\begin{proof}
We proceed by induction on $r$. If $r=1$ and $\textnormal{tdeg}(R_{1}(\cF))=0$, then $\cF$ consists of either a single column or a single row.
Therefore $R_2(\cF)=0$. Now assume $r \ge 2$ and that the statement is true for all $1 \le r' \le r-1$.
Suppose that $\textnormal{tdeg}(R_{r}(\cF))=0$, and denote by $\cF'$ the Ferrers diagram obtained from $\cF$ by deleting the last column.
We distinguish two cases.

\underline{Case 1:} $\textnormal{tdeg}(R_{r-1}(\cF'))=0$. By induction hypothesis we have $R_r(\cF')=0$, and so it must be that $R_{r+1}(\cF)=0$ as well.

\underline{Case 2:} $\textnormal{tdeg}(R_{r-1}(\cF'))> 0$.
By assumption there exists a placement~$C$ of $r$ rooks on $\cF$ such that $\textnormal{inv}(C,\cF)=0$.
Then all the rooks of $C$ must lie on $\cF'$ (as otherwise we would have $\textnormal{inv}(C',\cF)=0$, where $C'$ is obtained from $C$ by removing the rook
lying on $\cF \setminus \cF'$, and this contradicts $\textnormal{tdeg}(R_{r-1}(\cF'))> 0$).
Since $\textnormal{inv}(C,\cF)=0$, every dot in the last column of~$\cF$ is to the right of one of the $r$ rooks.
But this means that~$\cF$ has exactly $r$ non-empty rows.
This in turn implies, $R_{r+1}(\cF)=0$, as desired.
\end{proof}

\begin{theo}\label{T-IncreasDeg}
Let $\cF$ be a Ferrers diagram, and let $r \ge 2$. If $P_r(\cF)$ is not the zero polynomial, then
$$\deg(P_r(\cF))>\deg(P_{r-1}(\cF)).$$
\end{theo}

\begin{proof}
By Proposition \ref{prop:tdeg}, it suffices to show that $\textnormal{tdeg}(R_{r-1}(\cF)) > \textnormal{tdeg}(R_r(\cF))$.
Note first that by assumption, $R_r(\cF)\not=0$ and thus $\textnormal{tdeg}(R_r(\cF))<\infty$.
Thus the result is immediate if $R_{r-1}(\cF)=0$.

We henceforth assume that both $R_r(\cF)$ and $R_{r-1}(\cF)$ are non-zero, and hence $r\leq m$.
Let $t=\textnormal{tdeg}(R_{r-1}(\cF))$.
If $t=0$ then $R_{r}(\cF)$ must be the zero polynomial thanks to Lemma \ref{lem:tec}, and this contradicts our assumptions.
Therefore $t \ge 1$.
Let now~$C$ be a placement of $r-1$ non-attacking rooks on $\cF$ such that $\mbox{inv}(C,\cF)=t$.
Since $t\geq 1$, there is at least one dot $(i,j)\in\cF$ that has not been deleted (crossed out) by these rooks.
We need to consider various cases.

\underline{Case 1:} Suppose no rook of~$C$ is in row~$i$ and no rook is in column~$j$.
Then we may place a rook at position $(i,j)$ and obtain a placement of $r$ non-attacking rooks on~$\cF$.
Denoting this placement by~$C'$, we clearly have $\mbox{inv}(C',\cF)<\mbox{inv}(C,\cF)=t$.
By  definition of $R_r(\cF)$, this implies $\textnormal{tdeg}(R_r(\cF))<t=\textnormal{tdeg}(R_{r-1}(\cF))$, as desired.

\underline{Case 2:} Suppose there is a rook of~$C$ in row~$i$, but none in column~$j$.
Then this rook is at a position $(i,j')$ where  $j'>j$.
We may move this rook to position $(i,j)$ and obtain another placement, $C'$, of $r-1$ non-attacking rooks.
But then $\mbox{inv}(C',\cF)<\mbox{inv}(C,\cF)$, because we have to delete at least one more dot, namely the one at position $(i,j)$.
This contradicts the minimality of~$t$.
The case where there is a rook of~$C$ in column~$j$, but none in row~$i$, is symmetric and leads to a contradiction as well.

\underline{Case 3:} Suppose there is a rook in row~$i$ and a rook in column~$j$.
Let their positions be $(i,j')$ and $(i',j)$ for some $j'>j$ and $i'<i$.
Since these two positions are in~$\cF$, the same is true for the position $(i',j')$.
We may thus move these two rooks to the positions $(i,j)$ and $(i',j')$, respectively, and obtain another placement, $C'$, of $r-1$ non-attacking rooks.
Again, this leads to $\mbox{inv}(C',\cF)<\mbox{inv}(C,\cF)$ in contradiction to the minimality of~$t$.
This concludes the proof.
\end{proof}

\section{Partition-Preserving Maps}
\label{sec:pres}

In this short section we consider maps between matrix codes that preserve any of the partitions discussed in this paper.
We will see that these maps can easily be described when defined on the entire matrix space~$\Fnm$, but that there is no analogue of the classical MacWilliams
Extension Theorem~\cite{MacW62}.
%
The latter states that (1) the Hamming-weight-preserving maps $\F^n \longrightarrow \F^n$ are given by monomial matrices (i.e., matrices that have exactly one nonzero entry in each row and column), and (2)
for any code $\cC\leq\F^n$ each Hamming-weight-preserving map $\cC \longrightarrow \F^n$ extends to a Hamming-weight-preserving map on~$\F^n$.
In short, the Hamming isometries between codes in~$\F^n$ are monomial maps, and this fully describes these maps.
We refer to~\cite[Thm.~7.9.4]{HP03} for further details.
In this section, we study the analogous question for the rank, row-space, and pivot partition.
We do \textit{not} assume $m \le n$.

\begin{defi}\label{D-Preserving}
Let $\cC\leq\Fnm$ be a subspace and $f:\cC \longrightarrow \Fnm$ be a linear map.
\begin{enumerate}
\item $f$ is \textbf{rank-preserving} if $\rk(f(A))=\rk(A)$ for all $A\in\cC$.
\item $f$ is \textbf{row-space-preserving} if $\rs(f(A))=\rs(A)$ for all $A\in\cC$.
\item $f$ is \textbf{pivot-preserving} if $\piv(f(A))=\piv(A)$ for all $A\in\cC$.
\end{enumerate}
\end{defi}

Note that rank-preserving maps preserve the rank partition in the sense that $A$ and $f(A)$ are in the same block of~$\cP^\rk$ for all $A\in\cC$.
Similar reformulations are true for row-space-preserving  or pivot-preserving maps.
Thus we may call maps \textbf{partition-preserving} if they are of  the corresponding type above.

The question arises whether such maps can be described explicitly.
As in the classical situation with the Hamming distance, the simplest case occurs when the code~$\cC$ is the entire space.
In this case the question can be answered for all three partitions by making use of the following description of rank-preserving maps.
A first instance of this result has been proven by Hua~\cite{Hua51} (see also~\cite[Thm.~3.4]{Wan96}).
An elementary proof can be found in~\cite{MaMo59} by Marcus/Moyls.

\begin{theo}[\mbox{\hspace*{-.4em}\cite[Thm.~1]{MaMo59}}]\label{T-MaMo}
Let $f:\Fnm \longrightarrow \Fnm$ be a rank-preserving map.
Then there exist matrices $U\in\GL_n(\F)$ and $V\in\GL_m(\F)$ such that
\[
   f(A)=UAV \text{ for all }A\in\Fnm
\]
or, only in the case $n=m$,
\[
  f(A)=UA^\top V\text{ for all }A\in\F^{m\times m}.
\]
Clearly, any map~$f$ of such a form is rank-preserving.
\end{theo}

Let us briefly comment on this result for the case where $n\neq m$.
From the rank-preserving property it is clear that for every~$A$ in~$\Fnm$ there exist $U_A\in\GL_n(\F)$ and $V_A\in\GL_m(\F)$ such that  $f(A)=U_AAV_A$.
The strength of the above theorem lies in the fact that these matrices are \emph{global}, that is, they do not depend on~$A$.

Now the row-space-preserving and the pivot-preserving maps on~$\Fnm$ can be described easily.
Part~(1) below can also be proven with the aid of \cite[Thm.~4]{MaMa18}.
In~\cite{MaMa18}, the authors study (among other things) rank support spaces.
These are matrix spaces consisting of all matrices whose row space is contained in a fixed prescribed space.
In \cite[Thm.~4]{MaMa18} the maps preserving the ``rank support space property'' are characterized.
Since row-space-preserving maps are of this form, this result allows us to rule out immediately Case~2.\ in the proof of~(1) below.
However, since the proof of \cite[Thm.~4]{MaMa18} is quite long, we prefer to present our short, elementary proof based directly on Theorem~\ref{T-MaMo}.

\begin{cor}\label{C-RSPres}
Let $f:\Fnm \longrightarrow \Fnm$ be a linear map.
\begin{enumerate}
\item \label{C1} $f$ is row-space-preserving iff there exists $U\in\GL_n(\F)$ such that $f(A)=UA$ for all $A\in\Fnm$.
\item \label{C2} $f$ is pivot-preserving iff there exists $U\in\GL_n(\F)$ and $V\in\cU_m(\F)$ such that $f(A)=UAV$ for all $A\in\Fnm$, where
        $\cU_m(\F)=\{V\in\GL_m(\F)\mid V\text{ is upper triangular}\}$.
\end{enumerate}
\end{cor}

\begin{proof}
It is clear that maps of the form described in~(\ref{C1}), resp.~(\ref{C2}) are row-space-preserving, resp.\ pivot-preserving (see also Proposition~\ref{P-BasicsP}(\ref{p3})).
Let us now turn to the other implications.

1) Let~$f$ be row-space-preserving. Then~$f$ is also rank-preserving and we may apply Theorem~\ref{T-MaMo}.
\\
\underline{Case 1}: There exist $U\in\GL_n(\F)$ and $V\in\GL_m(\F)$ such that $f(A)=UAV$ for all $A\in\Fnm$.
Assume $V\neq \alpha I_m$ for any $\alpha\in\F^*$.
Then there exists $x\in\F^m$ such that $xV\not\in\text{span}\{x\}$.
Let $A\in\Fnm$ be such that
\[
     UA=\begin{pmatrix}x\\0\\\vdots\\0\end{pmatrix}.
\]
Then $\rs(A)=\rs(UA)=\text{span}\{x\}\neq\rs(UAV)$, a contradiction.
Thus $V=\alpha I_m$ for some $\alpha\in\F^*$ and $f(A)=(\alpha U)A$ for all $A\in\Fnm$, as desired.
\\
\underline{Case 2}: Let $m=n>1$ and suppose $U\in\GL_m(\F)$ and $V\in\GL_m(\F)$ are such that $f(A)=UA^\top V$ for all $A\in\F^{m\times m}$.
Write
\begin{equation}\label{e-V}
   V=\begin{pmatrix}V_1\\ \vdots\\ V_m\end{pmatrix}.
\end{equation}
Consider the standard basis matrices $E_{ij}\in\F^{m\times m}$ which have entry~$1$ at position $(i,j)$ and are zero elsewhere.
Then $\text{span}(e_j)=\rs(E_{ij})=\rs(UE_{ji}V)=\rs(E_{ji}V)=\text{span}(V_i)$ for all $i\in[m]$. This contradicts the invertibility of~$V$.
Hence this case does not occur.

2)  Let~$f$ be pivot-preserving. Then $f$ is also rank-preserving, and we may proceed as in (\ref{C1}).
\\
\underline{Case 1}: There exist $U\in\GL_n(\F)$ and $V\in\GL_m(\F)$ such that $f(A)=UAV$ for all $A\in\Fnm$.
Suppose $V=(v_{ij})$ is not upper triangular.
Then there exists a smallest~$j\in[m]$ and $i>j$ such that $v_{ij}\neq 0$.
With~$V$ as in~\eqref{e-V} we arrive at $(i)=\piv(E_{1i})=\piv(UE_{1i}V)=\piv(E_{1i}V)=\piv(V_i)=(j)$, which is a contradiction.
Thus~$V$ is upper triangular, as desired.
\\
\underline{Case 2}: Let $m=n>1$ and suppose $U\in\GL_m(\F)$ and $V\in\GL_m(\F)$ are such that $f(A)=UA^\top V$ for all $A\in\F^{m\times m}$.
Fix some $j>1$.
With~$V$ as in~\eqref{e-V} we obtain $(j)=\piv(E_{ij})=\piv(UE_{ji}V)=\piv(V_i)$ for all $i\in[m]$.
This means that the first $j-1$ columns of~$V$ are zero, a contradiction to the invertibility of~$V$.
Hence, again, this case cannot occur. \qedhere
\end{proof}

We conclude this paper with examples showing that for any of the partitions~$\cP^\rk,\,\cP^\rs,\,\cP^\piv$ the partition-preserving maps between codes in $\Fnm$ do not in general extend
to such maps on the entire space~$\Fnm$.
In other words, there is no analogue to the MacWilliams Extension Theorem.

\begin{exa}\label{E-NotExt}
Let $\F=\F_2$.
\begin{enumerate}
\item In~\cite[Ex.~2.9(a)]{BGL15} it is shown that for the code $\cC=\{(A\,|\,0)\in\F^{2\times 3}\mid A\in\F^{2\times 2}\}$ the rank-preserving map
        $f:\cC\longrightarrow\cC,\ (A\,|\,0)\longmapsto (A^\top|\,0)$
        does not extend to a rank-preserving map on $\F^{2\times 3}$.
\item In $\F^{3\times 3}$ consider the subset $\cC=\F[P]=\{0,I,P,\ldots,P^6\}$, where
        \[
           P=\begin{pmatrix}0&0&1\\1&0&1\\0&1&0\end{pmatrix}.
        \]
        Then~$P$ is the companion matrix of the primitive polynomial $x^3+x+1\in\F[x]$ and thus~$\cC$ is actually the field~$\F_8$.
        In particular, $A\in\GL_3(\F)$ for all $A\in\cC\setminus\{0\}$.
        As a consequence, the map
        \[
             f:\cC\longrightarrow\F^{3\times 3},\quad A\longmapsto A^\top
        \]
        is trivially row-space-preserving and pivot-preserving.
        We show that~$f$ does not extend to a pivot-preserving map on~$\F^{3\times 3}$.
        Assume to the contrary that it does extend.
        Then Corollary~\ref{C-RSPres}(2) tells us that there exist $U\in\GL_3(\F)$ and $V\in\cU_3(\F)$ such that $f(A)=UAV$ for all $A\in\F^{3\times 3}$.
        Since $I\in\cC$ we have $I=I^\top=f(I)=UIV$, and thus $U=V^{-1}$ is upper triangular.
        Now $P^\top=f(P)=UPU^{-1}$ implies $UP=P^\top U$. One easily verifies that no matrix $U\in\cU_3(\F)$ satisfies this identity.
        Hence~$f$ does not extend to a pivot-preserving map on $\F^{3\times 3}$ and thus also not to a row-space-preserving map.
\end{enumerate}
\end{exa}

\section*{Acknowledgement}
We would like to thank the reviewers for their very close and careful reading, their constructive comments, and for pointing out a gap in our original proof of Theorem~\ref{T-IncreasDeg}.

\bibliographystyle{abbrv}
\bibliography{Biblio_new}

\end{document}